\documentclass[journal]{IEEEtran}

\usepackage{cite}

\ifCLASSINFOpdf
   \usepackage[pdftex]{graphicx}

\else

\fi

\usepackage[cmex10]{amsmath}
\usepackage{amsfonts}
\usepackage{amsthm}

\usepackage{algorithm}
\usepackage{algpseudocode} 

\usepackage{subfigure}

\usepackage[utf8]{inputenc}
\usepackage[english]{babel}

\newtheorem{theorem}{Theorem}[section]

\newtheorem{lemma}[theorem]{Lemma}
\usepackage{bm}
\newcommand{\ma}[1]{\mathbf{ #1 }}         

\newcommand{\compl}{\mathbb{C}}        
\newcommand{\real}{\mathbb{R}}         

\usepackage{color}

\newcommand{\MainFigureSizes}{0.85\columnwidth}
\newcommand{\rev}[1]{{\color[rgb]{0,0,0}#1}}

\begin{document}
  
  
  

\title{\rev{Hardware Impairments Aware Transceiver Design for Bidirectional Full-Duplex MIMO OFDM Systems}}
		\author{Omid~Taghizadeh,~Vimal~Radhakrishnan,~Ali~Cagatay~Cirik,~\IEEEmembership{Member,~IEEE},~Rudolf Mathar,~\IEEEmembership{Senior~Member,~IEEE},~Lutz~Lampe~\IEEEmembership{Senior~Member,~IEEE} \vspace{-5mm}
			\IEEEcompsocitemizethanks{
		\IEEEcompsocthanksitem O. Taghizadeh,~V.~Radhakrishnan, and R. Mathar are with the Institute for Theoretical Information Technology, RWTH Aachen University, Aachen, 52074, Germany (email:  \{taghizadeh,~radhakrishnan,~mathar\}@ti.rwth-aachen.de).
					\IEEEcompsocthanksitem A.~C.~Cirik and L. Lampe are with the Department of Electrical and Computer Engineering, University of British Columbia, Vancouver, BC V6T 1Z4, Canada  (email: \{cirik,~lampe\}@ece.ubc.ca). 
					\IEEEcompsocthanksitem Part of this paper has been presented at the 28th IEEE Annual International Symposium on Personal, Indoor, and Mobile Radio Communications (PIMRC'17)~\cite{Tagh1710:Power}. 
							}}  
\maketitle

\begin{abstract}
In this paper we address the linear precoding and decoding design problem for a bidirectional orthogonal frequency-division multiplexing (OFDM) communication system, between two multiple-input multiple-output (MIMO) full-duplex (FD) nodes. \rev{The effects of hardware distortion as well as the channel state information error are taken into account. In the first step, we transform the available time-domain characterization of the hardware distortions for FD MIMO transceivers to the frequency domain, via a linear Fourier transformation. As a result, the explicit impact of hardware inaccuracies on the residual self-interference (RSI) and inter-carrier leakage (ICL) is formulated in relation to the intended transmit/received signals. Afterwards, linear precoding and decoding designs are proposed to enhance the system performance following the minimum-mean-squared-error (MMSE) and sum rate maximization strategies, assuming the availability of perfect or erroneous CSI. The proposed designs are based on the application of alternating optimization over the system parameters, leading to a necessary convergence. Numerical results indicate that the application of a distortion-aware design is essential for a system with a high hardware distortion, or for a system with a low thermal noise variance.}
\end{abstract}


\IEEEpeerreviewmaketitle
\vspace{-4mm}
\section{Introduction} \label{sec:into}

\IEEEPARstart{F}{ULL}-Duplex (FD) transceivers are known for their capability to transmit and receive at the same time and frequency, and hence have the potential to enhance the spectral efficiency~\cite{HBCJMKL:14}. Nevertheless, such systems suffer from the inherent self-interference \rev{(SI)} from their own transmitter. Recently, specialized self-interference cancellation (SIC) techniques, e.g., \cite{Bharadia:14, BMK:13, lee2016analog, 6517516}, have demonstrated an adequate level of isolation between transmit (Tx) and receive (Rx) directions to facilitate an FD communication and motivated a wide range of related studies, see, e.g., \cite{6719531, 6963403, 6656015, 7015632}. A common idea of such SIC techniques is to subtract the dominant part of the \rev{SI} signal, e.g., a line-of-sight (LOS) \rev{SI} path or near-end reflections, in the radio frequency (RF) analog domain so that the remaining signal can be further processed in the baseband, i.e., digital domain. Nevertheless, such methods are still far from perfect in a realistic environment mainly due to \emph{i)} aging and inherent inaccuracy of the hardware (analog) elements, as well as \emph{ii)} inaccurate channel state information (CSI) in the \rev{SI} path, due to noise and limited channel coherence time. In this regard, inaccuracy of the analog hardware elements used in subtracting the dominant \rev{SI} path in RF domain may result in severe degradation of SIC quality. This issue becomes more relevant in a realistic scenario, where unlike the demonstrated setups in the lab environment, analog components are prone to aging, temperature fluctuations, and occasional physical damage. Moreover, an FD link is vulnerable to CSI inaccuracy at the \rev{SI} path in environments with a small channel coherence time, see~\cite[Subsection~3.4.1]{JCKB:11}. A good example of such challenge is a high-speed vehicle that passes close to an FD device, and results in additional reflective \rev{SI} paths\footnote{Since the object is moving rapidly, the reflective paths are more difficult to be accurately estimated.}.            
                      
In order to combat the aforementioned issues, an FD transceiver may adapt its transmit/receive strategy to the expected nature of CSI inaccuracy, e.g., by directing the transmit beams away from the moving objects or operating in the directions with smaller impact of CSI error. Moreover, the accuracy of the transmit/receiver chain elements can be considered, e.g., by dedicating less power, or ignoring the chains with noisier elements in the signal processing. In this regard, a widely used model for the operation of a multiple-antenna FD transceiver is proposed in \cite{DMBS:12}, assuming a single carrier communication system, where CSI inaccuracy as well as the impact of hardware impairments are taken into account. A gradient-projection-based method is then proposed in the same work for maximizing the sum rate in an FD bidirectional setup. Building upon the proposed benchmark, a convex optimization design framework is introduced in \cite{Huberman2014, JTLH:12} by defining a price/threshold for the \rev{SI} power, assuming the availability of perfect CSI and accurate transceiver operation. While this approach reduces the design computational complexity, it does not provide a reliable performance for a scenario with erroneous CSI, particularly regarding the \rev{SI} path \cite{ZTH:132}. Consequently, the consideration of CSI and transceiver error in an FD bidirectional system is further studied in \cite{CRYL:15, Tagh1705:Worst} by maximizing the system sum rate, in \cite{6941786} by minimizing the sum mean-squared-error (MSE), and in \cite{7562572, Tagh1705:Sum} for minimizing the system power consumption under a required quality of service. \par 

The aforementioned works focus on modeling and design methodologies for single-carrier FD bidirectional systems, under frequency-flat channel assumptions.
In this regard, the importance of extending the previous designs for a multi-carrier (MC) system with a frequency selective channel is threefold.
Firstly, due to the increasing rate demand of the wireless services, and following the same rationale for the promotion of FD systems, the usage of larger bandwidths becomes necessary. This, in turn, invalidates the usual frequency-flat assumption and calls for updated design methodologies. 
Secondly, unlike the half-duplex (HD) systems where the operation of different subcarriers can be safely separated in the digital domain, an FD system is highly prone to the inter-carrier leakage (ICL) due to the impact of hardware distortions on the strong \rev{SI} channel\footnote{For instance, a high-power transmission in one of the subcarriers will result in a higher \rev{residual self-interference (RSI)} in all of the sub-channels due to, e.g., a higher quantization and power amplifier noise levels.}. This, in particular, calls for a proper modeling of the ICL as a result of non-linear hardware distortions for FD transceivers. \rev{And finally, the channel frequency selectivity shall be opportunistically exploited, by means of a joint design of the linear transmit and receive strategies at all subcarriers, in order to enhance the system performance.}   

\vspace{-4mm} \subsection{Related works on FD MC systems} \label{related_works}
In the early work by Riihonen \emph{et al.} \cite{6488955}, the performance of a combined analog/digital SIC scheme is evaluated for an FD orthogonal-frequency-division-multiplexing (OFDM) transceiver, taking into account the impact of hardware distortions, e.g., limited analog-to-digital convertor (ADC) accuracy. The problem of resource allocation and performance analysis for FD MC communication systems is then addressed in \cite{5449862,7270330,7504451,sun2016optimal,7454410,7194031,6832469}, however, assuming a single antenna transceiver.  
%
Specifically, an FD MC system is studied in \cite{5449862,7270330,7504451} in the context of FD relaying, in \cite{7454410, 7194031} and \cite{sun2016optimal} in the context of FD cellular systems with non-orthogonal multiple access (NOMA) capability, and in \cite{6832469} for rate region analysis of a hybrid HD/FD link. Moreover, an MC relaying system with hybrid decode/amplify-and-forward operation is studied in \cite{ng2012dynamic}, with the goal of maximizing the system sum rate via scheduling and resource allocation. However, in all of the aforementioned designs, the behavior of the residual \rev{SI} signal is modeled as a purely linear system. As a result, the impacts of the hardware distortions leading to ICL, as observed in \cite{6488955}, are neglected.  
\vspace{-4mm} \subsection{Contribution and paper organization}
In this paper we study a bidirectional FD MIMO OFDM system\footnote{\rev{The modeling and the obtained design frameworks can be applied also for any multi-carrier system with orthogonal waveforms, i.e., with zero intrinsic interference.}}, where the impacts of hardware distortions leading to imperfect SIC and ICL are taken into account. 
Our main contributions, together with the paper organization are summarized as follows:
\rev{\begin{itemize}[leftmargin=*]
\item In the seminal work by Day \emph{et al.} \cite{DMBS:12}, an FD MIMO transceiver is modeled considering the impacts of hardware distortions in transmit/receiver chains, which is then extensively used for the purpose of FD system design and performance analysis, e.g., \cite{DMBSR:12,CRYL:15, 8036658, XaZXMaXu:15, 7809157, 7562572, 8234646}. In the first step, we extend the available time-domain characterization of hardware distortions into an FD MIMO OFDM setup via a linear discrete Fourier transformation. The obtained frequency-domain characterization reveals the statistics of the RSI and ICL, in relation to the intended transmit/receive signals at each subcarrier. Please note that this is in contrast to the available prior works on FD MC systems \cite{5449862,7270330,7504451,sun2016optimal,7454410,7194031,6832469,ng2012dynamic}, where ICL is neglected and \rev{RSI} signal is modeled via a purely linear system.


\item Building on the obtained characterization, linear transmit/receive strategies are proposed in order to enhance the system performance. In Section~\ref{sec:WMMSE}, an alternating quadratic convex program (QCP), denoted as AltQCP, is proposed in order to obtain a minimum weighted MSE transceiver design. The known weighted-minimum-MSE (WMMSE) method \cite{CACC:08} is then utilized to extend the AltQCP framework for maximizing the system sum rate. For both algorithms, a monotonic performance improvement is observed at each step, leading to a necessary convergence.   

\item In Section~IV, we extend the studied system to an asymmetric OFDM FD bidirectional setup, where an FD transceiver with a large antenna array simultaneously communicates with multiple single-antenna FD transceivers. The extended scenario is particularly relevant, both due to the recent advances in building FD massive MIMO transceivers \cite{7575689} as well as the signified impact of hardware distortions due to the lower per-element cost (e.g., low resolution quantization~\cite{7932472}). An algorithm for joint power and subcarrier allocation is then proposed, following the successive inner approximation (SIA) framework \cite{marks1978technical}, with a guaranteed convergence to a solution satisfying Karush–Kuhn–Tucker (KKT) conditions. 

\item In Section~\ref{sec:WMMSE_CSI_Error} the proposed design in Section~\ref{sec:WMMSE} is extended by also taking into account the impact of CSI error. In particular, a worst-case MMSE design is proposed as an alternating semi definite program (SDP), denoted as AltSDP. Similar to the previous methods, a monotonic performance improvement is observed at each step, leading to a necessary convergence. Moreover, a methodology to obtain the most destructive CSI error matrices is proposed. This is done by converting the resulting non-convex quadratic problem into a convex program, in order to facilitate worst-case performance analysis under CSI error.     
\end{itemize}   }
Numerical simulations show that the application of a distortion-aware design is essential, as transceiver accuracy degrades, and ICL becomes a dominant factor. 
\vspace{-4mm} \subsection{Mathematical Notation} 
Throughout this paper, column vectors and matrices are denoted as lower-case and {upper-case} bold letters, respectively. {Mathematical expectation, trace}, inverse, determinant, transpose, conjugate {and} Hermitian transpose are denoted by $ \mathbb{E}\{\cdot\}, \; {\text{tr}}(\cdot), \; (\cdot)^{-1}\; |\cdot|, \; (\cdot)^{ T},\; (\cdot)^{*}$ {and} $(\cdot)^{ H},$ respectively. The Kronecker product is denoted by $\otimes$. The identity matrix with dimension $K$ is denoted as ${\ma I}_K$ and ${\text{vec}}(\cdot)$ operator stacks the elements of a matrix into a vector. $\ma{0}_{m \times n}$ represents an all-zero matrix with size $m \times n$. $\| \cdot \|_{2}$ and $\|\cdot\|_{{F}}$ {respectively represent the Euclidean and Frobenius norms}. $\text{diag}(\cdot)$ returns a diagonal matrix by putting the off-diagonal elements to zero. $\left\lfloor \mathbf{A}_i \right \rfloor_{i=1,\ldots,K}$ denotes a tall matrix, obtained by stacking the matrices $\mathbf{A}_i,~i=1,\ldots, K$. $\mathcal{R}\{\ma{A}\}$ represents the range (column space) of the matrix $\ma{A}$. The set $\mathbb{F}_K$ is defined as $\{1,\ldots,K\}$. The set of real, positive real, and complex numbers are respectively denoted as $\mathbb{R}, \mathbb{R}^+ , \compl$. 

\vspace{-4mm}
\section{System Model}\label{sec:model}

{{\footnotesize{\begin{table}[!t]  
	\renewcommand{\arraystretch}{1}
  \caption{\rev{Used symbols}}  \vspace{-2mm} \label{tab_notations}
  \centering
  \rev{\begin{tabular}[t]{|c||l|}
   \hline	 
	 $k,i,l$    &  index of a subcarrier, communication direction, \\
	   &  and a transmit/receive chain  \\
	 $ \mathbb{I}, \mathbb{V}, \mathbb{U} $   & set of comm. directions, precoder and decoder matrices   	\\
	 $N_i, M_i, d_i$  &  number of transmit and receive antennas and data streams\\
	 ${\ma{s}}_{i}^k (\tilde{\ma{s}}_{i}^k)$  &  transmitted (estimated) data symbol\\
	 ${\ma{U}_{i}}^k ({\ma{V}_{i}}^k)$  &  linear decoder (precoder) matrix \\
	 ${\ma{y}}_{i}^k (\tilde{\ma{y}}_{i}^k)$  &  received signal before (after) SI cancellation\\
	${\ma{H}}_{ij}^{k}, \tilde{\ma{H}}_{ij}^{k} $  & the exact, and estimated CSI matrix \\
	 ${\ma{\Delta}}_{ij}^{k}, \mathbb{D}_{ij}^{k} $  &  CSI error, and the set of feasible CSI errors\\	
	 $\zeta_{ij}^k, \ma{D}_{ij}^k $  &  radius and shaping matrix for the feasible CSI error region  \\
	 $\ma{e}_{\text{r}, i}^k (\ma{e}_{\text{t}, i}^k)$  & receiver (transmitter) distortion over the subcarrier $k$ \\
	 $\ma{\Theta}_{\text{rx},i} (\ma{\Theta}_{\text{trx},i})$  & diagonal matrix of receive (transmit) distortion coefficients \\
	 $\boldsymbol{\nu}_{i}^k, \ma{\Sigma}_{i}^k$  & collective residual SI plus noise signal, and its covariance \\
	 $\ma{n}_{i}^k, \sigma_{i,k}^2$  &  additive thermal noise and its variance\\
	 $\ma{u}_{i}^k (\ma{v}_{i}^k)$  & undistorted received (transmitted) signal\\
	 $\ma{x}_{i}^{k}, P_i$  &  transmit signal, and the maximum transmit power\\
	\hline
  \end{tabular}}  \vspace{-2mm} 
\end{table} }}
A bidirectional {OFDM} communication between two {MIMO} {FD} transceivers is considered. Each communication direction is associated with $N_i$ transmit and $M_i$ receive antennas, where $i \in \mathbb{I}$, and $\mathbb{I}:=\{1,2\}$ represents the set of the communication directions. The desired channel in the communication direction $i$ and subcarrier $k \in \mathbb{F}_K$ is denoted as $\ma{H}_{ii}^{k} \in \compl^{M_i \times N_i}$ where $K$ is the number of subcarriers. The interference channel from $i$ to $j$-th communication direction is denoted as $\ma{H}_{ji}^{k}\in \compl^{M_j \times N_i}$. All channels are quasi-static\footnote{It indicates that the channel is constant in each communication frame, but may vary from one frame to another frame.}, and frequency-flat in each subcarrier. 

The transmitted signal in the direction $i$, subcarrier $k$ is formulated as 
\begin{align} \label{model:transmitsignal}
\ma{x}_{i}^{k} = \underbrace{\ma{V}_{i}^{k} \ma{s}_{i}^{k}}_{=:\ma{v}_{i}^{k}} + \ma{e}_{\text{t}, i}^{k} , \;\; \sum_{k\in \mathbb{F}_K} \mathbb{E} \left\{ \|\ma{x}_{i}^{k}\|_2^2 \right\} \leq P_i,  
\end{align}     
where $\ma{s}_{i}^{k} \in \compl^{d_i}$ and $\ma{V}_{i}^{k} \in \compl^{N_i \times d_i}$ respectively represent the vector of the data symbols and the transmit precoding matrix, and $P_i \in \real^+$ imposes the maximum affordable transmit power constraint. The number of the data streams in each subcarrier and in direction $i$ is denoted as $d_i$, and $\mathbb{E}\{\ma{s}_{i}^{k}{\ma{s}_{i}^{k}}^H\} = \ma{I}_{d_i}$. Moreover, $\ma{v}_{i}^{k} \in \compl^{N_i}$ represents the desired signal to be transmitted, where $\ma{e}_{\text{t},i}^{k}$ models the inaccurate behavior of the transmit chain elements, i.e, transmit distortion, see Subsection~\ref{distortion_signal_explaination} for more details. \par 
The received signal at the destination can be consequently written as 
\begin{align} \label{model_received_beforecancellation}
\ma{y}_{i}^k = \underbrace{ \ma{H}_{ii}^k \ma{x}_{i}^k +  \ma{H}_{ij}^k \ma{x}_{j}^k + \ma{n}_{i}^k}_{=:\ma{u}_{i}^k} + \ma{e}_{\text{r}, i}^k,
\end{align}  
where $\ma{n}_{i}^k \sim \mathcal{CN}\left( \ma{0}_{M_i},  \sigma_{i,k}^2 \ma{I}_{M_i} \right)$ is the additive thermal noise. Similar to the transmit signal model, $\ma{e}_{\text{r}, i}^k$ represents the receiver distortion and models the inaccuracies of the receive chain elements.
The \emph{known}, i.e., distortion-free, part of the {SI} is then subtracted from the received signal, employing an {SIC} scheme. This is formulated as  
\begin{align} \label{model:rx_signal}
\tilde{\ma{y}}_{i}^k :&= {\ma{y}}_{i}^k - {\ma{H}}_{ij}^k \ma{V}_{j}^k\ma{s}_{j}^k   = \ma{H}_{ii}^k \ma{V}_{i}^k \ma{s}_{i}^k + {\boldsymbol{\nu}}_{i}^k, 
\end{align} 
where $\tilde{\ma{y}}_{i}^k$ is the received signal in direction $i$ and subcarrier $k$, after {SIC}. Moreover, the aggregate interference-plus-noise term is denoted as $\boldsymbol{\nu}_{i}^k \in \compl^{M_i}$, where 
\begin{align} \label{model:ri_signal}
{\boldsymbol{\nu}}_{i}^k & = \ma{H}_{ij}^k \ma{e}_{\text{t},j}^k + \ma{H}_{ii}^k \ma{e}_{\text{t},i}^k + \ma{e}_{\text{r},i}^k  + \ma{n}_{i}^k, \;\; j \neq i.
\end{align}
Finally, the estimated data vector is obtained at the receiver as 
\begin{align} \label{model:ri_estimatedsignal}
\tilde{\ma{s}}_{i}^k = \left( {\ma{U}_{i}}^k \right)^H \tilde{\ma{y}}_{i}^k, 
\end{align}
where $\ma{U}_{i}^k \in \compl^{M_i \times d_i }$ is the linear receive filter. 
\vspace{-2mm} 
\subsection{Limited dynamic range in an FD OFDM system} \label{distortion_signal_explaination}
\rev{In the seminal work by Day \emph{et al.} \cite{DMBS:12}, a model for the operation of an {FD} {MIMO} transceiver is given, relying on the experimental results on the impact of hardware distortions \cite{MITRX:05, MITTX:98,MITTX:08, FD_ExperimentDrivenCharact}.
In this regard, the inaccuracy of the transmit chain elements, e.g., {DAC} error, {PA} and oscillator phase noise, are jointly modeled for each antenna as an additive distortion, and written as $x_l (t) = v_l (t) + {e}_{\text{t},l} (t) $, see Fig.~\ref{TransceiverAccuracyModel}, such that 
\begin{align} 
& {e}_{\text{t},l} (t) \sim \mathcal{CN} \Big( 0, {\kappa}_{l} \mathbb{E} \big\{ \left| v_l (t) \right|^2 \big\} \Big), \label{eq_model_distortion_stat_1} \\
& {e}_{\text{t},l} (t) \bot {v}_{l} (t),\;\; {e}_{\text{t},l} (t) \bot {e}_{\text{t},l^{'}} (t), \;\; {e}_{\text{t},l} (t) \bot {e}_{\text{t},l} (t^{'}), \nonumber \\ & \;\;\;\;\;\;\;\;\;\;\;\;\;\;\;\;\;\;\;\;\;\;\;\;\;\;\;\;\;\;\;\;\;\;\;\;\;\;\;\;\;\;\;\;\;\;\; l\neq {l^{'}} \in \mathbb{L}_T, \;\; t\neq {t^{'}}, \label{eq_model_distortion_stat_2}
\end{align}
\rev{please see \cite[Section~II.~B,C]{DMBS:12},~\cite[Section~II.~C,D]{DMBSR:12}, \cite{CRYL:15, 8036658, XaZXMaXu:15, 7809157, 7562572} for a similar distortion characterization for FD transceivers\footnote{\rev{It is worth mentioning that the accuracy of the above-mentioned modeling varies for different implementations of FD transceivers, depending on the complexity and the used SIC method. In this regard, the statistical independence of distortion elements defined in (iii) and (iv) also hold for an advanced implementation of an FD transceiver, assuming a high signal processing capability. This is since any correlation structure in the distortion signal can be exploited and removed in order to reduce the \rev{RSI} via advanced signal processing, see \cite[Subsection~3.2]{BMK:13}. However, the linear dependence of the remaining distortion signal variance to the desired signal strength varies for different SIC implementations, and should be estimated separately for each transceiver.}}. In the above arguments, $t$ denotes the instance of time, and $v_l$, $x_l$, and $e_{\text{t},l} \in \compl$ are respectively the baseband time-domain representation of}} the intended transmit signal, the actual transmit signal, and the additive transmit distortion at the $l$-th transmit chain. The set $\mathbb{L}_T$ represents the set of all transmit chains. Moreover, ${\kappa}_{l} \in \real^+$ represents the distortion coefficient for the $l$-th transmit chain, relating the collective power of the distortion signal, over the active spectrum, to the intended transmit power. \par
In the receiver side, the combined effects of the inaccurate hardware elements, i.e., {ADC} error, {AGC} and oscillator phase noise, are presented as additive distortion terms and written as $y_l (t) =  u_l (t) + {e}_{\text{r},l} (t) $ such that 
\begin{align} 
& {e}_{\text{r},l} (t) \sim \mathcal{CN} \Big( 0, {\beta}_{l} \mathbb{E} \big\{  \left| u_l (t) \right|^2 \big\} \Big), \label{eq_model_distortion_stat_3}\\
& {e}_{\text{r},l} (t) \bot {u}_{l} (t),\;\; {e}_{\text{r},l} (t) \bot {{e}_{\text{r},l^{'}}} (t), \;\; {e}_{\text{r},l} (t) \bot {e}_{\text{r},l} (t^{'}), \nonumber \\ & \;\;\;\;\;\;\;\;\;\;\;\;\;\;\;\;\;\;\;\;\;\;\;\;\;\;\;\;\;\;\;\;\;\;\;\;\;\;\;\;\;\;\;\;\;  l\neq {l^{'}} \in \mathbb{L}_R, \;\; t \neq {t^{'}}, \label{eq_model_distortion_stat_4}
\end{align}
where $u_l$, $e_{\text{r},l}$, and $y_l \in \compl$ {are respectively the baseband representation of} the intended (distortion-free) received signal, additive receive distortion, and the received signal from the $l$-th receive antenna. The set $\mathbb{L}_R$ represents the set of all receive chains. Similar to the transmit chain characterization, ${\beta}_{l}  \in \real^+$ is the distortion coefficient for the $l$-th receive chain, see Fig.~\ref{TransceiverAccuracyModel}. \rev{For each communication block, the frequency domain representation of the sampled time domain signal is obtained as
\begin{align} \label{FD_P2P_FrequencyDomain_Tx}
& x_l^k = \frac{1}{\sqrt{K}} \sum_{m = 0}^{K-1} x_l (m T_{\text{s}}) e^{-\frac{j2 \pi mk}{K}}  =   \nonumber \\ & \underbrace{\frac{1}{\sqrt{K}} \sum_{m = 0}^{K-1} v_l (m T_{\text{s}}) e^{- \frac{j2 \pi mk}{K}} }_{=:v_l^k} + \underbrace{\frac{1}{\sqrt{K}} \sum_{m = 0}^{K-1} e_{\text{t},l} (m T_{\text{s}}) e^{- \frac{j2 \pi mk}{K}} , }_{=:e_{\text{t},l}^k} 
\end{align} \hspace{-0mm}
and 
\begin{align} \label{FD_P2P_FrequencyDomain_Rx}
& y_l^k = \frac{1}{\sqrt{K}} \sum_{m = 0}^{K-1} y_l (m T_{\text{s}}) e^{-\frac{j 2 \pi m k}{K}}  = \nonumber \\ & \underbrace{\frac{1}{\sqrt{K}} \sum_{m = 0}^{K-1} u_l (m T_{\text{s}}) e^{- \frac{j2 \pi mk}{K}} }_{=:u_l^k} + \underbrace{\frac{1}{\sqrt{K}} \sum_{m = 0}^{K-1} e_{\text{r},l} (m T_{\text{s}}) e^{- \frac{j2 \pi mk}{K}} , }_{=:e_{\text{r},l}^k}
\end{align} 
where $T_{\text{s}}$ is the sampling time, and $K T_{\text{s}}$ is the OFDM block duration prior to the cyclic extension, see \cite{cho2010mimo} for a detailed discussion on OFDM technology.  
\begin{lemma} \label{OFDM_Distortion}
The impact of hardware distortions in the frequency domain is characterized as  
\begin{align}\hspace{-0mm}
& {e}_{\text{t},l}^k  \sim \mathcal{CN} \left( 0, \frac{{\kappa}_{l}}{K} \sum_{m \in \mathbb{F}_K} \mathbb{E} \left\{ \left| v_l^m \right|^2 \right\} \right),\;\; e_{\text{t},l}^k \bot v_l^k, \;\; e^k_{\text{t},l} \bot {e^k_{\text{t},{l{}'}}},  \\
& {e}_{\text{r},l}^k  \sim \mathcal{CN} \left( 0, \frac{{\beta}_{l}}{K} \sum_{m \in \mathbb{F}_K} \mathbb{E} \left\{ \left| u_l^m \right|^2 \right\} \right),\;\; e_{\text{r},l}^k \bot u_l^k, \;\; e^k_{\text{r},l} \bot {e^k_{\text{r},{l{}'}}}, 
\end{align} 
transforming the statistical independence, as well as the proportional variance properties from the time domain. 
\end{lemma}
\begin{proof}
See the Appendix.
\end{proof}

\par
The above lemma indicates that the distortion signal variance at each subcarrier, relates to the total distortion power at the corresponding chain, indicating the impact of ICL. This can be interpreted as a variance-dependent thermal noise, where the temporal independence of signal samples results in a flat power spectral density over the active communication bandwidth. In this part we consider a general framework where the transmit (receive) distortion coefficients are not necessarily identical for all transmit (receive) chains belonging to the same transceiver, i.e., different chains may hold different accuracy due to occasional damage and aging. This assumption is relevant in practice since it enables the design algorithms to reduce communication task, e.g., transmit power, on the chains with noisier elements. Following Lemma~\ref{OFDM_Distortion}, the statistics of the distortion terms, introduced in (\ref{model:transmitsignal}), (\ref{model_received_beforecancellation}) can be inferred as    
\begin{align} 
& \ma{e}_{\text{t},i}^k \sim \mathcal{CN} \left( \ma{0}_{N_i}, \ma{\Theta}_{\text{tx},i}  \sum_{k \in \mathbb{F}_K} \text{diag} \left( \mathbb{E} \left\{  \ma{v}_i^k {\ma{v}_i^k}^H \right \} \right) \right), \label{eq_model_distortion_stat_res_tx}  \\  
& \ma{e}_{\text{r},i}^k \sim \mathcal{CN} \left( \ma{0}_{M_i},   \ma{\Theta}_{\text{rx},i} \sum_{k \in \mathbb{F}_K}  \text{diag} \left( \mathbb{E} \left\{  \ma{u}_i^k {\ma{u}_i^k}^H \right \} \right) \right),  \label{eq_model_distortion_stat_res_rx} 
\end{align}
where $\ma{\Theta}_{\text{tx},i} \in \real^{N_i \times N_i}$ ($\ma{\Theta}_{\text{rx},i} \in \real^{M_i \times M_i}$) is a diagonal matrix including distortion coefficients ${\kappa}_{l}/K$ (${\beta}_{l}/K$) for the corresponding chains}\footnote{A simpler mathematical presentation can be obtained by assuming the same transceiver accuracy over all antennas, similar to \cite{DMBS:12, DMBSR:12}. In such a case, the defined diagonal matrices can be replaced by a scalar.}.  
\begin{figure}[!t] 
    \begin{center}
        \includegraphics[angle=0, width=0.85\columnwidth]{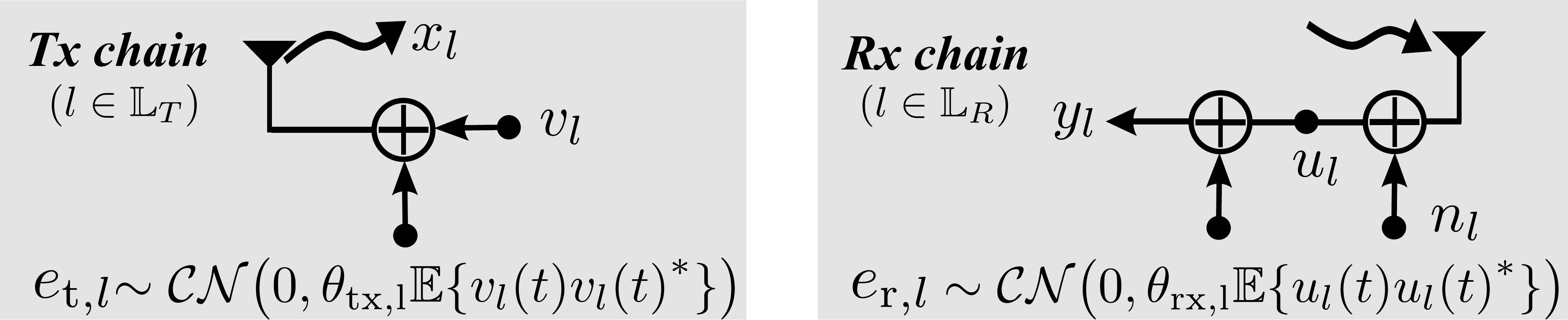}
    \caption{{{Limited dynamic range is modeled by injecting additive distortion terms at each transmit or receive chain. $e_{\text{t},l}$ and $e_{\text{r},l}$ denote the distortion terms, and $n_l$ \rev{represents} the additive thermal noise.}}} \label{TransceiverAccuracyModel} \vspace{-7mm} 
    \end{center}  
\end{figure} 

Via the application of (\ref{eq_model_distortion_stat_res_tx})-(\ref{eq_model_distortion_stat_res_rx}) on (\ref{model:ri_signal}), the covariance of the received collective interference-plus-noise signal is obtained as
{\small{\begin{align} \label{eq_model_aggregate_interference_covariance_CSI_Perfect}
& \ma{\Sigma}_{i}^k  := \mathbb{E} \left\{ {\boldsymbol{\nu}}_{i}^k  {{\boldsymbol{\nu}}_{i}^k}^H \right\} \nonumber \\ 
 & \approx \sum_{j \in \mathbb{I}} \ma{H}_{ij}^k \ma{\Theta}_{\text{tx},j} \text{diag} \left( \sum_{l \in \mathbb{F}_K } \ma{V}_j^l{\ma{V}_j^l}^H \right) {\ma{H}_{ij}^k}^H  + \sigma_{i,k}^2 \ma{I}_{M_i} \nonumber \\
& + \ma{\Theta}_{\text{rx},i}   \text{diag} \bigg( \sum_{l\in\mathbb{F}_K}  \bigg(   \sigma_{i,l}^2 \ma{I}_{M_i}  + \sum_{ j \in \mathbb{I} } \ma{H}_{ij}^l \ma{V}_j^l{\ma{V}_j^l}^H {\ma{H}_{ij}^l}^H   \bigg) \bigg) ,
\end{align}}}
where $\ma{\Sigma}_{i}^k \in \compl^{M_i \times M_i}$ is obtained considering $0 \leq {\beta}_{l} \ll 1$, $0 \leq {\kappa}_{l} \ll 1$, and hence ignoring the terms containing higher orders of the distortion coefficients in (\ref{eq_model_aggregate_interference_covariance_CSI_Perfect}).
\subsection{Remarks}
\begin{itemize}
\item In this section, we have assumed the availability of perfect {CSI} and focused on the impact of non-linear transceiver distortions. This assumption is relevant for the scenarios with stationary channel, e.g., a backhaul directive link with zero mobility \cite{7306534}, where an adequately long training sequence can be applied, see \cite[Subsection~III.A]{DMBS:12}. The impact of the {CSI} inaccuracy is later addressed in Section~\ref{sec:WMMSE_CSI_Error}.   
\item As expected, the role of the distortion signals on the \rev{RSI}, including the resulting {ICL}, is evident from (\ref{eq_model_aggregate_interference_covariance_CSI_Perfect}). It is the main goal of the remaining parts of this chapter to incorporate and evaluate this impact on the design of the defined {MC} system. 
\end{itemize}

\vspace{-4mm}
\section{Linear Transceiver Design for Multi-Carrier Communications} \label{sec:WMMSE}

Via the application of $\ma{V}_i^k$ and $\ma{U}_i^k$, as the linear transmit precoder and receive filters, the mean-squared-error (MSE) matrix of the defined system is calculated as
\begin{align} \label{MSE_Matrix}
\ma{E}_i^k :&= \mathbb{E} \left\{ \left( \tilde{\ma{s}}_i^k - {\ma{s}}_i^k \right) \left( \tilde{\ma{s}}_i^k - {\ma{s}}_i^k \right)^H \right\} \nonumber \\  &= \left( {\ma{U}_i^k}^H \ma{H}_{ii}^k \ma{V}_i^k - \ma{I}_{d_i} \right)\left( {\ma{U}_i^k}^H \ma{H}_{ii}^k \ma{V}_i^k - \ma{I}_{d_{i}} \right)^H  \nonumber  \\ & \;\;\;\; + {\ma{U}_i^k}^H  \ma{\Sigma}_i^k {\ma{U}_i^k},  
\end{align}
where $\ma{\Sigma}_i^k$ is given in (\ref{eq_model_aggregate_interference_covariance_CSI_Perfect}). In the following we propose two design strategies for the defined system, proposing an alternating QCP framework. 
 
\subsection{Weighted MSE minimization via Alternating QCP (AltQCP)} \label{sec:MWMSE}
An optimization problem for minimizing the weighted sum MSE is written as 
\begin{subequations}  \label{eq:global_opt_problem_MWMSE}
\begin{align}
\underset{\mathbb{V},\mathbb{U}}{ \text{min}} \;\;   & \sum_{i \in \mathbb{I}}\sum_{k \in \mathbb{F}_K}  \text{tr} \left({\ma{S}_i^k} \ma{E}_i^k\right)  \label{eq:global_opt_problem_MWMSE_a} \\
{\text{s.t.}} \;\; & \text{tr}\bigg( \left( \ma{I}_{N_i} + K\ma{\Theta}_{\text{tx},i}\right) \sum_{l\in \mathbb{F}_K} \ma{V}_i^l{\ma{V}_i^l}^H \bigg) \leq P_i, \;\; \forall i \in \mathbb{I},  \label{eq:global_opt_problem_MWMSE_b}
\end{align} 
\end{subequations}
where $\mathbb{X}:= \{\ma{X}_{i}^k, \; \forall i \in \mathbb{I}, \; \forall k \in \mathbb{F}_K\}$, with $\mathbb{X} \in \{\mathbb{U}, \mathbb{V}\}$, and (\ref{eq:global_opt_problem_MWMSE_b}) represents the transmit power constraint. It is worth mentioning that the application of ${\ma{S}_i^k} \succ 0$, as a weight matrix associated with $\ma{E}_i^k$ is two-folded. Firstly, it may appear as a diagonal matrix, emphasizing the importance of different data streams and different users. Secondly, it can be applied as an auxiliary variable which later relates the defined weighted MSE minimization to a sum-rate maximization problem, see Subsection~\ref{sec:perf_CSI_Rate}. \par 
It is observed that (\ref{eq:global_opt_problem_MWMSE}) is not a jointly convex problem. Nevertheless, it holds a QCP structure separately over the sets $\mathbb{V}$ and $\mathbb{U}$, in each case when other variables are fixed. In this regard, the objective (\ref{eq:global_opt_problem_MWMSE_a}) can be decomposed over $\mathbb{U}$ for different communication directions, and for different subcarriers. The optimal minimum MSE (MMSE) receive filter can be hence calculated in closed form as 
\begin{align} \label{wmmse_U_mmse}
\ma{U}_{i,\text{mmse}}^k = \left( \ma{\Sigma}_i^k + \ma{H}_{ii}^k\ma{V}_i^k{\ma{V}_i^k}^H{\ma{H}_{ii}^k}^H \right)^{-1} {\ma{H}_{ii}^k} {\ma{V}_i^k}.
\end{align} 
Nevertheless, the defined problem is coupled over $\ma{V}_{i}^k$, due to the impact of inter-carrier leakage, as well as the power constraint (\ref{eq:global_opt_problem_MWMSE_b}). The Lagrangian function, corresponding to the optimization (\ref{eq:global_opt_problem_MWMSE}) over $\mathbb{V}$ is expressed as 
\begin{align} \label{Lagrangian_WMMSE}     
\mathcal{L} & \left( \mathbb{V}, \boldsymbol{\iota}\right)  :  = \sum_{i \in \mathbb{I}} \bigg( \iota_i { \mathcal{P}_i }\left( \mathbb{V} \right) + \sum_{k \in \mathbb{F}_K}  \text{tr} \left({\ma{S}_i^k} \ma{E}_i^k\right)     \bigg), \\ &  \mathcal{P}_i \left( \mathbb{V} \right):=  - P_i +   \text{tr}\bigg( \left( \ma{I}_{N_i} + K\ma{\Theta}_{\text{tx},i}\right) \sum_{l\in \mathbb{F}_K} \ma{V}_i^l{\ma{V}_i^l}^H \bigg),   
\end{align}
where $ \boldsymbol{\iota}:= \{\iota_i,\; i\in\mathbb{I}\}$ is the set of dual variables. The dual function, corresponding to the above Lagrangian is defined as 
\begin{align} 
\mathcal{F}\left( \boldsymbol{\iota} \right)  : & = \underset{\mathbb{V}}{\text{min}} \;\;   \mathcal{L} \left( \mathbb{V}, \boldsymbol{\iota}\right) \label{WMMSE_dualfunction} 
\end{align} 
where the optimal ${\ma{V}_i^k}$ is obtained as
\begin{align} \label{WMMSE_ClosedForm_V}
& {\ma{V}_i^k}^\star = \left( {\ma{J}_i^k} + \iota_i \left( \ma{I}_{N_i} + K\ma{\Theta}_{\text{tx},i}\right) + {\ma{H}_{ii}^k}^H {\ma{U}_i^k} {\ma{S}_i^k} {\ma{U}_i^k}^H {\ma{H}_{ii}^k }  \right)^{-1} \nonumber \\ & \;\;\;\;\;\;\;\;\;\;\;\;\;\;\;\;\;\;\;\;\;\;\;\;\;\;\;\;\;\;\;\;\;\;\;\;\;\;\;\;\;  \times {\ma{H}_{ii}^k}^H {\ma{U}_i^k}{\ma{S}_i^k},
\end{align} 
and
\begin{align} 
{\ma{J}_i^k} :&= \sum_{l \in \mathbb{F}_K} \sum_{j \in \mathbb{I}} \bigg( {\ma{H}_{ji}^k}^H \text{diag} \left( {\ma{U}_j^l} {\ma{S}_j^l} {\ma{U}_j^l}^H  {\ma{\Theta}_{\text{rx},j}}  \right) {\ma{H}_{ji}^k}  \nonumber \\  &  \;\;\;\;  +   \text{diag} \left({\ma{H}_{ji}^l}^H  {\ma{U}_j^l}  {\ma{S}_j^l} {\ma{U}_j^l}^H  {\ma{H}_{ji}^l}  {\ma{\Theta}_{\text{tx},i}}  \right)  
\bigg).   
\end{align}
Due to the convexity of the original problem (\ref{eq:global_opt_problem_MWMSE}) over $\mathbb{V}$, the defined dual problem is a concave function over $\boldsymbol{\iota}$, with $\mathcal{P}_i(\mathbb{V})$ as a subgradient, see \cite[Eq.~(6.1)]{bertesekas1999nonlinear}. As a result, the optimal $\boldsymbol{\iota}$ is obtained from the maximization 
\begin{align} 
\boldsymbol{\iota}^\star = \underset{\boldsymbol{\iota} \geq 0}{\text{argmax}} \; \mathcal{F}\left( \boldsymbol{\iota} \right),
\end{align} 
following a standard subgradient update,~\cite[Subsection~6.3.1]{bertesekas1999nonlinear}. \par
Utilizing the proposed optimization framework, the alternating optimization over $\mathbb{V}$ and $\mathbb{U}$ is continued until a stable point is obtained. Note that due to the monotonic decrease of the objective in each step, and the fact that (\ref{eq:global_opt_problem_MWMSE_a}) is non-negative and hence bounded from below, the defined procedure leads to a necessary convergence. Algorithm~\ref{alg:MWMSE} defines the necessary optimization steps.
{\small{\begin{algorithm}[H] 
 \small{	\begin{algorithmic}[1] 
\State{$\ell \leftarrow  {0}  ;  \;\;\;\; \text{(set iteration number to zero)}$}
\State{$\mathbb{V} \leftarrow \text{right~singular~matrix~initialization,~see~\cite[Appendix~A]{5585631}}$ }
\State{$\mathbb{U} \leftarrow \text{solve~(\ref{wmmse_U_mmse})}$ }
\Repeat 
\State{$\ell \leftarrow  \ell + 1 $}
\State{$\mathbb{V} \leftarrow \text{solve (\ref{WMMSE_ClosedForm_V}) or QCP (\ref{eq:global_opt_problem_MWMSE}), with fixed} \; \mathbb{U}$}

\State{$\mathbb{U} \leftarrow \text{solve (\ref{wmmse_U_mmse}) or QCP (\ref{eq:global_opt_problem_MWMSE}) with fixed} \; \mathbb{V}$}

\Until{$\text{a stable point, or maximum number of $\ell$ reached}$}

\State{\Return$\left\{\mathbb{U},\mathbb{V}\right\}$}
  \end{algorithmic} } 
 \caption{\small{Alternating QCP (AltQCP) for weighted MSE minimization} } \label{alg:MWMSE}
\end{algorithm}  }  	}

\subsection{Weighted MMSE (WMMSE) design for sum rate maximization} \label{sec:perf_CSI_Rate}
Via the utilization of $\ma{V}_{i}^k$ as the transmit precoders, the resulting communication rate for the $k$-th subcarrier and for the $i$-th communication direction is written as 
\begin{align} \label{model:rate_formulation}
I_i^k = \text{log}_2 \left| \ma{I}_{d_i} + {{\ma{V}}_{i}^k}^H{{\ma{H}}_{ii}^k}^H \big({\ma{\Sigma}}_{i}^k\big)^{-1} {\ma{H}}_{ii}^k{\ma{V}}_{i}^k  \right|, 
\end{align}
where ${\ma{\Sigma}}_{i}^k$ is defined in (\ref{eq_model_aggregate_interference_covariance_CSI_Perfect}). The sum rate maximization problem can be hence presented as
\begin{align} \label{eq:model_optimization}
\underset{\mathbb{V}}{ \text{max}} \;\; & \;\;\sum_{i \in \mathbb{I}}\sum_{k \in \mathbb{F}_K} \omega_i I_i^k , \;\; {\text{s.t.}} \;\; \text{(\ref{eq:global_opt_problem_MWMSE_b})}. 
\end{align} 
\rev{where $\omega_i \in \real^+$ is the weight associated with the communication direction $i$}. The optimization problem (\ref{eq:model_optimization}) is intractable in the current form. In the following we propose an iterative optimization solution, following the WMMSE method \cite{CACC:08}. \\
Via the application of the MMSE receive linear filters from (\ref{wmmse_U_mmse}), the resulting MSE matrix is obtained as 
\begin{align} \label{wmmse_E_mmse}
{\ma{E}_{ i,{\text{mmse}} }^k} = \left(  \ma{I}_{d_i} + {\ma{V}_i^k}^H {\ma{H}_{ii}^k}^H \left({\ma{\Sigma}_i^k}\right)^{-1} \ma{H}_{ii}^k \ma{V}_i^k \right)^{-1}.
\end{align}
By recalling (\ref{model:rate_formulation}), and upon utilization of $\ma{U}_{i,\text{mmse}}^k$, we observe the following useful connection to the rate function 
\begin{align} \label{wmmse_wmmseEquivalent}
I_i^k = - \text{log}_2 \left| \ma{E}_{i,\text{mmse}}^k \right|,
\end{align}
which facilitates the decomposition of rate function via the following lemma, see also \cite[Eq.~(9)]{CACC:08}.   
\vspace{-1mm}
\begin{lemma} \label{lemma_logdetE}
Let $\ma{E} \in \compl^{d \times d}$ be a positive definite matrix. The maximization of the term $\normalfont{-\text{log} \left|\ma{E} \right|}$ is equivalent to the maximization
\begin{align}
{\normalfont{\underset{\ma{E}, \ma{S}}{ \text{max}} - \text{tr}\left( \ma{S} \ma{E} \right) +  \text{log} \left|\ma{S} \right| + d,}}
\end{align} 
where $\ma{S} \in \compl^{d \times d}$ is a positive definite matrix, and we have 
\begin{align} \label{W_opt_perfectCSI}
\ma{S} = \ma{E}^{-1} ,
\end{align}
at the optimality.
\end{lemma} \vspace{-1mm}
\begin{proof}
See \cite[Lemma~2]{JPKR:11}.
\end{proof}
By recalling (\ref{wmmse_wmmseEquivalent}), and utilizing Lemma~\ref{lemma_logdetE}, the original optimization problem over $\mathbb{V}$ can be equivalently formulated as
\begin{align} \label{eq:global_opt_problem_rate}
\underset{\mathbb{V},\mathbb{U},\mathbb{S}}{ \text{max}} \;\; &  \sum_{i \in \mathbb{I}} \omega_i \sum_{k \in \mathbb{F}_K}\bigg(\text{log}\left|\ma{S}_i^k\right|  + d_i  - \text{tr} \left({\ma{S}_i^k} \ma{E}_i^k\right)\bigg) \;\; {\text{s.t.}} \;\; \text{(\ref{eq:global_opt_problem_MWMSE_b})}, 
\end{align} 
where $\mathbb{S}:= \{\ma{S}_{i}^k \succ 0, \; \forall i \in \mathbb{I}, \; \forall k \in \mathbb{F}_K\}$.  
The obtained optimization problem (\ref{eq:global_opt_problem_rate}) is not a jointly convex problem. Nevertheless, it is a QCP over $\mathbb{V}$ when other variables are fixed, and can be obtained with a similar structure as for (\ref{eq:global_opt_problem_MWMSE}). Moreover, the optimization over $\mathbb{U}$ and $\mathbb{S}$ is respectively obtained from (\ref{wmmse_U_mmse}), and (\ref{W_opt_perfectCSI}) as $\ma{S}_i^k = {\ma{E}_i^k}^{-1}$. This facilitates an alternating optimization where in each step the corresponding problem is solved to optimality, see Algorithm~\ref{alg:WMMSE_rate}. The defined alternating optimization steps results in a necessary convergence due to the monotonic increase of the objective in each step, and the fact that the eventual system sum rate is bounded from above. 
{\small{ \begin{algorithm}[H] 
 \small{	\begin{algorithmic}[1] 
\State{$\text{Algorithm~1, Steps 1-2}  \;\;\;\; \text{(initialization)}$}
\Repeat 
\State{$\text{Algorithm~1, Steps 5-7}$}
\State{$\mathbb{S} \leftarrow \ma{S}_i^k = \left( {\ma{E}_i^k} \right)^{-1}$ }
\Until{$\text{a stable point, or maximum number of $\ell$ reached}$}

\State{\Return$\left\{\mathbb{V} \right\}$}
  \end{algorithmic} } 
 \caption{\small{AltQCP-WMMSE design for sum rate maximization} } \label{alg:WMMSE_rate}
\end{algorithm} }}

\rev{\section{Bidirectional FD Massive MIMO Systems: Joint Power and Subcarrier Allocation}
In this part, we extend the studied system into an asymmetric setup, where an FD transceiver equipped with a large antenna array (e.g., a basestation) performs a bidirectional communication with multiple FD single-antenna nodes (e.g., users). Thanks to the FD capability and multi-user beamforming, the communication at different directions can flexibly coexist on shared subcarriers, improving the spectral efficiency, or can be accommodated on different subcarriers in order to control the interference. Please note that the impact of hardware impairments is known to be significant for a system with a large antenna array, due to the lower per-element cost, e.g., low resolution ADC and DAC \cite{7932472}. This signifies the role of the characterization in Lemma~\ref{OFDM_Distortion} regarding the impact of hardware impairments for an FD MIMO OFDM system. In order to extend the defined setup to an asymmetric one, we denote the set of communication directions from (to) the users to (from) the massive MIMO transceiver as $\mathbb{I}_{UL}$ ($\mathbb{I}_{DL}$), such that $\mathbb{I} = \mathbb{I}_{UL} \cup \mathbb{I}_{DL}$. Moreover, the lower-case notations\footnote{\rev{The channel dimensions are accordingly obtained as $\ma{H}^k_{ij}\in{\compl^{1}}$ when $i \in \mathbb{I}_{UL}$, $j \in \mathbb{I}_{DL}$, $\ma{h}^k_{ii}\in{\compl^{\tilde{M}}} ({\compl^{1 \times \tilde{N}}})$ when $i \in \mathbb{I}_{UL} (\mathbb{I}_{DL})$, and $\ma{H}^k_{ji}\in{\compl^{\tilde{M} \times \tilde{N}}}$ when $i \in \mathbb{I}_{DL} $, $j \in \mathbb{I}_{UL}$. $\tilde{N}$ ($\tilde{M}$) represent the number of transmit (receive) antennas at the massive MIMO transceiver.}} ($\tilde{\ma{f}}^k_{i}$) ${\ma{f}}^k_{i}$ and $\ma{u}^k_{i}$ are used to represent the (normalized) transmit and receive linear filters\footnote{\rev{Due to the properties of the large antenna arrays, the transmit precoder and receive filters are usually chosen via a maximum ratio transmission/combining (MRT/MRC) strategy \cite{7932472}, a projection to the null-space of the SI channel~\cite{7575689} or via a joint user and SI spatial zero-forcing~\cite{6876818, 7208847}, resulting in a different performance-complexity tradeoff.}}. Moreover, we have ${\ma{f}}^k_{i} = \tilde{\ma{f}}^k_{i} \sqrt{p_{i,k}}$ where $p_{i,k}$ denotes the transmit power. In this part, we perform a joint subcarrier and power allocation with the goal of maximizing the system sum rate. An upper bound on the achievable information rate is obtained as 
\begin{align} \label{MMIMO_RateUpperBound}
R_{i,k}^{\text{UB}} = \gamma_0 \text{log}_2 \bigg( 1 + \frac{ \left|\left(\ma{u}^k_{i}\right)^H  \ma{h}_{ii}^k  \tilde{\ma{f}}^m_{i} \right|_2^2 p_{i,k} }{  \sigma_{i,k}^2  + \sum_{j \in \mathbb{I}} \sum_{m \in \mathbb{F}_K} \gamma_{ij}^{km} p_{j,m}  } \bigg)
\end{align}
where $0 < \gamma_0 < 1 $ indicates the portion of the frame duration dedicated to data communication, and  
{\small{\begin{align} \label{MMIMO_gamma}
 \gamma_{ij}^{km} :&= \underbrace{ \delta_{ij}^{kl} \left|\left(\ma{u}^k_{i}\right)^H \ma{H}^m_{ij}  \tilde{\ma{f}}^m_{j}  \right|^2  }_{\text{co-channel interference}} +  \nonumber \\   & +\underbrace{\left(\ma{u}^k_{i}\right)^H \ma{H}^k_{ij} \ma{\Theta}_{\text{tx},j} \text{diag} \left( \tilde{\ma{f}}^m_{j} \left(\tilde{\ma{f}}^m_{j}\right)^H \right) \left(\ma{H}^k_{ij}\right)^H \ma{u}^k_{i}}_{\text{transmitter distortion}}  \nonumber \\
& + \underbrace{ \left(\ma{u}^k_{i}\right)^H \ma{\Theta}_{\text{rx},i} \text{diag} \left( \ma{H}^m_{ij}  \tilde{\ma{f}}^m_{j} \left(\tilde{\ma{f}}^m_{j}\right)^H  \left(\ma{H}^m_{ij}\right)^H \right) \ma{u}^k_{i} }_{\text{receiver distortion}}, 
\end{align}}}    
where $\delta_{ij}^{km}=0$ if $k\neq m$ or $j \in \mathbb{I}_{DL}, i \in \mathbb{I}_{UL}$ and otherwise $\delta_{ij}^{km}=1$. Please note that the given upper bound in (\ref{MMIMO_RateUpperBound}) is obtained similar to \cite{7208847} assuming an accurate CSI, please see Section~\ref{sec:WMMSE_CSI_Error} for the consideration of CSI error. It is worth mentioning the impact of hardware distortions on the \rev{RSI}, as well as the ICL is evident from (\ref{MMIMO_gamma})\footnote{In particular to a massive MIMO transceiver, where low-resolution quantization is used, the distortion coefficient $\kappa_l$ in $\ma{\Theta}_{\text{tx},i}$ (and similarly $\beta_l$ in $\ma{\Theta}_{\text{rx},i}$) is obtained as $\kappa_l~\text{[dB]} = -6.02 b_l$, where $b_l$ is the number of quantization bits at the chain $l$.}. The optimization problem for maximizing the system sum rate is formulated as 
\begin{subequations} \label{mMIMO_optimizationproblem}                             
\begin{align} 
\underset{{p_{i,k}} \geq 0}{\text{max}} \;\; & \;\;\sum_{i \in \mathbb{I}} \omega_i \sum_{k \in \mathbb{F}_K} R_{i,k}^{\text{UB}} , \;\; \\ {\text{s.t.}} \;\;  &\sum_{k \in \mathbb{F}_K} p_{i,k} \leq P_i, i \in \mathbb{I}_{UL}, \sum_{i \in \mathbb{I}_{DL}} \sum_{k \in \mathbb{F}_K} p_{i,k} \leq P_i, i \in \mathbb{I}_{DL}.\label{mMIMO_optimizationproblem_b}
\end{align}
\end{subequations}
It can be observed that (\ref{mMIMO_optimizationproblem}) is not a jointly convex optimization problem. However, it falls into the class of smooth difference-of-convex (DC) optimization problems. In this regard, we propose an iterative optimization, following the successive inner approximation (SIA) framework \cite{marks1978technical} which is proven to converge to a point satisfying KKT optimality conditions. Let $p_{i,k,0}$ be a feasible transmit power value. Then, employing the first order Taylor’s approximation on the concave terms, the value of $R_{i,k}^{\text{UB}}$ is lower-bounded as 
{\small{\begin{align} \label{mMIMO_UB_LB}
R_{i,k}^{\text{UB}} & \geq \gamma_0 \text{log}_2 \bigg( \left\|\ma{h}_{ii}^k\right\|_2^2 p_{i,k} + \sum_{j \in \mathbb{I}} \sum_{m \in \mathbb{F}_K} \gamma_{ij}^{km} p_{j,m}   + \sigma_{i,k}^2  \bigg)  \nonumber \\ &- \gamma_0 \text{log}_2 \bigg( \sum_{j \in \mathbb{I}} \sum_{m \in \mathbb{F}_K}  \gamma_{ij}^{km} p_{j,m,0} + \sigma_{i,k}^2  \bigg) \nonumber \\   & - \frac{\gamma_0\sum_{j \in \mathbb{I}} \sum_{m \in \mathbb{F}_K}  \gamma_{ij}^{km} \left( p_{j,m} - p_{j,m,0} \right) }{\text{log}(2) \sum_{j \in \mathbb{I}} \sum_{m \in \mathbb{F}_K} \gamma_{ij}^{km} p_{j,m,0}   + \sigma_{i,k}^2} =: \overline{R_{i,k}^{\text{UB}}},
\end{align}}}\\ 
where $\overline{R_{i,k}^{\text{UB}}}$ is a jointly concave function over $p_{i,k}$, facilitating an iterative update where in each iteration the convex problem $\underset{{p_{i,k}} \geq 0}{\text{max}} \;\; \sum_{i \in \mathbb{I}} \sum_{k \in \mathbb{F}_K} \overline{R_{i,k}^{\text{UB}}}$ s.t.~(\ref{mMIMO_optimizationproblem_b}) is solved to the optimality. The proposed iterative update is continued until a stable solution is obtained. It can be observed that $\overline{R_{i,k}^{\text{UB}}}$ represents a tight and global lower bound to ${R_{i,k}^{\text{UB}}}$, with a shared slope at the point of approximation $p_{j,m,0}$\footnote{\rev{This is directly concluded for a first-order Taylor's approximation on any smooth convex function~\cite{BV:04}.}}. As a result, the proposed iterative update follows the requirements set in \cite[Theorem~1]{marks1978technical}, with a proven convergence to a solution satisfying the KKT conditions. 


}

\vspace{-4mm}
\section{Robust Design with Imperfect CSI} \label{sec:WMMSE_CSI_Error}
In many realistic scenarios the CSI matrices can not be estimated or communicated accurately due to the limited channel coherence time as a result of, e.g., reflections from a moving object, or due to dedicating limited resource on the training/feedback process. This issue becomes more significant in \rev{an FD} system, due to the strong \rev{SI} channel which calls for dedicated silent times for tuning and training process, see~\cite[Subsection~III.A]{DMBS:12}. In particular, the impact of CSI error on the defined MC FD system is three-fold. Firstly, similar to the usual HD scenarios, it results in the erroneous equalization in the receiver, as the communication channels are not accurately known. Secondly, it results in an inaccurate estimation of the received signal from the \rev{SI} path, and thereby degrades the SIC quality. Finally, due to the CSI error, the impact of the distortion signals may not be accurately known, as the statistics of the distortion signals directly depend on the channel situation. In this part we extend the proposed designs in Section~\ref{sec:WMMSE} where the aforementioned uncertainties, resulting from CSI error, are also taken into account. 

\vspace{-1mm} \subsection{Norm-bounded CSI error}
In this part we update the defined system model in Section~\ref{sec:model} to the scenario where the CSI is known erroneously. In this respect we follow the so-called deterministic model \cite{WP:09}, where the error matrices are not known but located, with a sufficiently high probability, within a known feasible error region\footnote{\rev{The feasible error region can be obtained from the statistical distribution of the true CSI values, as a minimum radius ball or ellipsoid containing the true CSI values with a desired confidence probability, or via the knowledge of the CSI quantization strategy, in case the CSI error is dominated by feedback quantization.}}. This is expressed as 
\begin{align} \label{eq:model_channel_eror}
\ma{H}_{ij}^{k} =  \tilde{\ma{H}}_{ij}^{k} + {\ma{\Delta}}_{ij}^{k}, \;\; {\ma{\Delta}}_{ij}^{k} \in \mathbb{D}_{ij}^{k}, \;\; i,j\in\mathbb{I},
\end{align}
and
\begin{align} \label{eq:model_channel_eror}
\mathbb{D}_{ij}^{k} :=  \left\{ {\ma{\Delta}}_{ij}^{k} \; \big{\vert} \;  \| {\ma{D}}_{ij}^{k} {\ma{\Delta}}_{ij}^{k} \|_F \leq \zeta_{ij}^{k} \right\}, \;\; \forall i,j \in\mathbb{I}, \; k \in \mathbb{F}_{K},
\end{align}
where $\tilde{\ma{H}}_{ij}^{k}$ is the estimated channel matrix and ${\ma{\Delta}}_{ij}^{k}$ represents the channel estimation error. Moreover, ${\ma{D}}_{ij}^{k}\succeq 0$ and $\zeta_{ij}^{k} \geq 0$ jointly define a feasible ellipsoid region for ${\ma{\Delta}}_{ij}^{k}$ which generally depends on the noise and interference statistics, and the used channel estimation method. For further elaboration on the used error model see \cite{WP:09} and the references therein.  \\
The aggregate interference-plus-noise signal at the receiver is hence updated as  
\begin{align} \label{model:ri_signal_CSIError}
{\boldsymbol{\nu}}_{i}^k & = \ma{H}_{ij}^k \ma{e}_{\text{t},j}^k + \ma{H}_{ii}^k \ma{e}_{\text{t},i}^k + \ma{e}_{\text{r},i}^k + \ma{\Delta}_{ij}^k \ma{V}_{j}^k \ma{s}_{j}^k + \ma{n}_{i}^k, \;\; j \neq i \in \mathbb{I},
\end{align}
where $\ma{\Sigma}_i^k$, representing the covariance of ${\boldsymbol{\nu}}_{i}^k$, is expressed in (\ref{eq_model_aggregate_interference_covariance_CSIError}).  

\begin{figure*}[!ht]
\normalsize
{\small{\begin{align} \label{eq_model_aggregate_interference_covariance_CSIError}
\ma{\Sigma}_{i}^k & = \ma{\Delta}_{ij}^k \ma{V}_j^k{\ma{V}_j^k}^H {\ma{\Delta}_{ij}^k}^H \hspace{-1mm}
  + \hspace{-1mm} \sum_{ j \in \mathbb{I} } \hspace{-0.5mm} \ma{H}_{ij}^k \ma{\Theta}_{\text{tx},j} \text{diag} \left( \hspace{-0.5mm} \sum_{l \in \mathbb{F}_K } \ma{V}_j^l{\ma{V}_j^l}^H \hspace{-1mm} \right) {\ma{H}_{ij}^k}^H \hspace{-1mm} + \hspace{-1mm} \ma{\Theta}_{\text{rx},i}   \text{diag} \bigg( \hspace{-0.5mm} \sum_{l\in\mathbb{F}_K}  \hspace{-1mm} \bigg(   \hspace{-1mm} \sigma_{i,l}^2 \ma{I}_{M_i} \hspace{-1mm} + \hspace{-1mm} \sum_{ j \in \mathbb{I} } \ma{H}_{ij}^l \ma{V}_j^l{\ma{V}_j^l}^H {\ma{H}_{ij}^l}^H  \hspace{-1mm} \bigg) \bigg) \hspace{-0.5mm} + \hspace{-0.5mm} \sigma_{i,k}^2 \ma{I}_{M_i}.
\end{align} }}
\hrulefill
\vspace*{-0mm}
\end{figure*}


\vspace{-1mm} \subsection{Alternating SDP (AltSDP) for worst-case MSE minimization} \label{WMMSE_CSI_Error}
An optimization problem for minimizing the worst-case MSE under the defined norm-bounded CSI error is written as 
\begin{align}
\underset{\mathbb{V},\mathbb{U}}{ \text{min}} \;\; \underset{\mathbb{C}}{\text{max}} \;\;  & \sum_{i \in \mathbb{I}}\sum_{k \in \mathbb{F}_K}  \text{tr} \left({\ma{S}_i^k} \ma{E}_i^k\right),   \nonumber \\ 
 {\text{s.t.}} \;\;\; &  \text{(\ref{eq:global_opt_problem_MWMSE_b})}, \;\; \ma{\Delta}_{ij}^{k} \in \mathbb{D}_{ij}^{k}, \;\;  \forall i,j \in \mathbb{I}, \;\;  k \in \mathbb{F}_K,   \label{eq:global_opt_problem_MWMSE_CSIError}
\end{align}  
where $\mathbb{C} := \{\ma{\Delta}_{ij}^k, \; \forall i,j \in \mathbb{I} , \;\forall k \in \mathbb{F}_K\}$, and $\ma{E}_{i}^k$ is obtained from (\ref{MSE_Matrix}) and (\ref{eq_model_aggregate_interference_covariance_CSIError}). Note that the above problem is intractable, due to the inner maximization of quadratic convex objective over $\mathbb{C}$, which also invalidates the observed convex QCP structure in (\ref{eq:global_opt_problem_MWMSE}). In order to formulate the objective into a tractable form, we calculate
{\small{\begin{align}
& \sum_{k\in \mathbb{F}_K} \text{tr} \left({\ma{S}_i^k} \ma{E}_i^k \right) \nonumber \\ & =  \sum_{k\in \mathbb{F}_K} \Bigg{(}
\left\| {\ma{W}_i^k}^H \left( {\ma{U}_i^k}^H \ma{H}_{ii}^k \ma{V}_i^k - \ma{I}_{d_i} \right) \right\|_{{F}}^2 \nonumber \\
 & \;\;\;\; + \left\| {\ma{W}_i^k}^H {\ma{U}_i^k}^H \ma{\Delta}_{i3-i}^k \ma{V}_{3-i}^k   \right\|_{{F}}^2 + \sigma_{i,k}^2 \left\|  {\ma{W}_i^k}^H {\ma{U}_i^k}^H\right\|_{{F}}^2 \nonumber   \\
 & \;\;\;\; + \sum_{j \in \mathbb{I} } \sum_{l\in \mathbb{F}_{N_j}}  \sum_{m\in \mathbb{F}_K}   \left\| {\ma{W}_i^k}^H {\ma{U}_i^k}^H \ma{H}_{ij}^k \left(\ma{\Theta}_{\text{tx},j}\right)^{\frac{1}{2}}  \ma{\Gamma}_{N_j}^l  \ma{V}_j^m   \right\|_{{F}}^2  \nonumber \\
 & \;\;\;\; + \sum_{j \in \mathbb{I} } \sum_{l\in \mathbb{F}_{M_i}}  \sum_{m\in \mathbb{F}_K}   \left\| {\ma{W}_i^k}^H {\ma{U}_i^k}^H \left(\ma{\Theta}_{\text{rx},i}\right)^{\frac{1}{2}}  \ma{\Gamma}_{M_i}^l  \ma{H}_{ij}^m\ma{V}_j^m   \right\|_{{F}}^2  \nonumber \\
 & \;\;\;\;+ \left\|  {\ma{W}_i^k}^H {\ma{U}_i^k}^H   \left( \ma{\Theta}_{\text{rx},i} \sum_{q\in\mathbb{F}_K} \sigma_{i,q}^2  \right)^{\frac{1}{2}}  \right\|_{F}^2 \Bigg{)}  \label{error_representation_fro_norm} \\
& = {\sum_{j \in \mathbb{I} }  \sum_{k \in \mathbb{F}_K } \Big\| \ma{c}_{ij}^k + \ma{C}_{ij}^k \text{vec}\left( \ma{\Delta}_{ij}^k \right) \Big\|_{2}^2 } ,  \label{quadratic_error_representation_final}
\end{align}  }}
where $\ma{\Gamma}_{M}^l$ is an $M \times M$ zero matrix except for the $l$-th diagonal element equal to $1$. In the above expressions $\ma{W}_{i}^k = \left( \ma{S}_{i}^k \right)^{\frac{1}{2}}$, and 
{\small{\begin{align}\label{calculate_c_ij}
& \ma{c}_{ij}^k {:=}  \nonumber \\   &\left[ \begin{array}{c} \delta_{ij} \text{vec} \left( {\ma{W}_i^k}^H \left( {\ma{U}_i^k}^H \tilde{\ma{H}}_{ij}^k\ma{V}_j^k  -  \ma{I}_{d_j \delta_{ij}}   \right) \right)  \\   \left\lfloor \text{vec} \left( {\ma{W}_i^k}^H {\ma{U}_i^k}^H \tilde{\ma{H}}_{ij} \left( \ma{\Theta}_{\text{tx},j} \right)^{\frac{1}{2}}  \ma{\Gamma}_{N_j}^l   \ma{V}_j^m \right) \right\rfloor_{l \in \mathbb{F}_{N_j} , m \in \mathbb{F}_K}   \\ \left\lfloor \text{vec} \left( {\ma{W}_i^m}^H {\ma{U}_i^m}^H  \left( \ma{\Theta}_{\text{rx},i}  \right)^{\frac{1}{2}} \ma{\Gamma}_{M_{i}}^l \tilde{\ma{H}}_{ij}^k  \ma{V}_j^k \right) \right\rfloor_{l \in \mathbb{F}_{M_i}, m \in \mathbb{F}_K } \\   \delta_{ij} \text{vec} \left( {\ma{W}_i^k}^H {\ma{U}_i^k}^H  \left( \sigma_{i,k}^2 \ma{I}_{M_i} + \ma{\Theta}_{\text{rx},i} \sum_{m \in \mathbb{F}_K} \sigma_{i,m}^2 \right)^{\frac{1}{2}} \right) \end{array} \right], 
\end{align}
\begin{align}\label{calculate_C_ij}
& \ma{C}_{ij}^k {:=}  \nonumber \\ &\left[ \begin{array}{c}  {\ma{V}_j^k}^T \otimes  \left( {\ma{W}_i^k}^H {\ma{U}_i^k}^H  \right)  \\
\left\lfloor  \left(  \left( \ma{\Theta}_{\text{tx},j} \right)^{\frac{1}{2}}  \ma{\Gamma}_{N_j}^{l} \ma{V}_j^m \right)^T \otimes \left( {\ma{W}_i^k}^H {\ma{U}_i^k}^H \right) \right\rfloor_{l \in \mathbb{F}_{N_j} , m\in \mathbb{F}_K}  \\ \left\lfloor {\ma{V}_j^k}^T \otimes \left( {\ma{W}_i^m}^H {\ma{U}_i^m}^H \left( \ma{\Theta}_{\text{rx},i} \right)^{\frac{1}{2}} \ma{\Gamma}_{M_i}^l \right) \right\rfloor_{l \in \mathbb{F}_{M_i}, m \in \mathbb{F}_K } \\ \ma{0}_{M_i d_i \times M_i N_i}  \end{array} \right],
\end{align}  }}
where $\delta_{ij}$ is the Kronecker delta where $\delta_{ij}=1$ for $i=j$ and zero otherwise. Moreover we have $\ma{c}_{ij}^k \in \compl^{ \tilde{d}_{ij} \times 1}$, $\ma{C}_{ij}^k \in \compl^{ \tilde{d}_{ij} \times {M_i N_j}}$ such that  
\begin{align}\label{length_d}
& \tilde{d}_{ij} :=  d_id_j \left( 1 + K \left( N_j + M_i \right) \right) + d_i M_i .
\end{align}
Please note that (\ref{error_representation_fro_norm}) is obtained by recalling (\ref{MSE_Matrix}) and (\ref{eq_model_aggregate_interference_covariance_CSIError}) and the known matrix equality \cite[Eq.~(516)]{MCB:08}, and (\ref{calculate_c_ij})-(\ref{calculate_C_ij}) are calculated via the application of \cite[Eq.~(496),~(497)]{MCB:08}. \par

By applying the Schur's complement lemma on the epigraph form of the quadratic norm (\ref{quadratic_error_representation_final}), i.e., $ \big\| \ma{c}_{ij}^k + \ma{C}_{ij}^k \text{vec}\left( \ma{\Delta}_{ij}^k \right) \big\|_{2}^2 \leq \tau_{ij}^k$, the optimization problem (\ref{eq:global_opt_problem_MWMSE_CSIError}) is equivalently written as    
\begin{subequations} \label{global_opt_problem_MWMSE_CSIError}
\begin{align}
& \underset{\mathbb{V},\mathbb{U}, \mathbb{T}} {\text{min}} \;\; \underset{\mathbb{C}}{\text{max}} \;\;   \sum_{i \in \mathbb{I}}\sum_{k \in \mathbb{F}_K}  \tau_{ij}^k,  \;\;\; {\text{s.t.}} \;\;\; \text{(\ref{eq:global_opt_problem_MWMSE_b})}, \;\; \| \ma{b}_{ij}^k \|_{F} \leq \zeta_{ij}^k,    \\    
&  \left[\hspace{-0mm}\begin{array}{cc} 0 &\hspace{-0mm}  {{{\ma{b}}_{ij}^k}}{}^H {{\tilde{\ma{D}}_{ij}^k}}{}^{H} {{\ma{C}_{ij}^k}}{}^H   \\   {\ma{C}_{ij}^k}{\tilde{\ma{D}}_{ij}}^k {\ma{b}}_{ij}^k  & \hspace{-0mm}\ma{0}_{\tilde{d}_{ij} \times \tilde{d}_{ij}} \end{array} \hspace{-0mm}\right] \hspace{-0mm}+\hspace{-0mm} \left[\begin{array}{cc} \tau_{ij}^k & {{{\ma{c}}_{ij}^k}}{}^H \\   {\ma{c}}_{ij}^k &  \ma{I}_{\tilde{d}_{ij}} \end{array} \right] \succeq 0, 
\end{align}
\end{subequations} 
where $\mathbb{T}:= \{\tau_{ij}^k, \; \forall i,j \in \mathbb{I}, \; \forall k \in \mathbb{F}_K\}$ and 
\begin{align}
\tilde{\ma{D}}_{ij}^k & := \ma{I}_{N_j} \otimes \left({\ma{D}_{ij}^k}\right)^{-1}, \\
\tilde{\ma{\Delta}}_{ij}^k & := \ma{D}_{ij}^k {\ma{\Delta}}_{ij}^k, \;\; \ma{b}_{ij}^k  := \text{vec}\left(\tilde{\ma{\Delta}}_{ij}^k \right), 
\end{align}
are defined for notational simplicity. The problem (\ref{global_opt_problem_MWMSE_CSIError}) is still intractable, due to the inner maximization. The following lemma converts this structure into a tractable form. 

\vspace{-2mm}
\begin{lemma} \label{petersen}
Generalized Petersen's sign-definiteness lemma: Let $\ma{Y} = \ma{Y}^H$, and $\ma{X},\ma{P},\ma{Q}$ are arbitrary matrices with complex valued elements. Then we have
\begin{align}
\ma{Y} \succeq \ma{P}^H \ma{X} \ma{Q} + \ma{Q}^H \ma{X}^H \ma{P}, \;\; \forall \ma{X} \; : \; \| \ma{X}\|_F \leq \zeta, 
\end{align} 
if and only if 
\begin{align}
\exists \lambda \geq 0, \; \left[\begin{array}{cc} \ma{Y} - \lambda \ma{Q}^H\ma{Q} & - \zeta \ma{P}^H \\   - \zeta \ma{P} & \lambda \ma{I} \end{array} \right] \succeq 0.
\end{align} 
\end{lemma}
\vspace{-2mm} \begin{proof}
See~\cite[Proposition~2]{EM:04}, \cite{khlebnikov2008petersen}.
\end{proof}
By choosing the matrices in Lemma~\ref{petersen} such that $\ma{X} = {\ma{b}}_{ij}^k$, $\ma{Q} = \left[ -1, \; \ma{0}_{1 \times \tilde{d}_{ij} }  \right]$ and  
\begin{align} \label{CSI_Error_ChoosingMatrices_Lemma_Peterson}
\ma{Y} =  \left[\begin{array}{cc} \tau_{ij}^k & {{\ma{c}}_{ij}^k}{}^H \\   {\ma{c}}_{ij}^k & \ma{I}_{\tilde{d}_{ij}} \end{array} \right], \ma{P} = \left[\begin{array}{cc} \ma{0}_{{M_i N_j} \times 1 }, \;\; {\tilde{\ma{D}}_{ij}^k}{}^H {{\ma{C}}_{ij}^k}{}^H \end{array} \right], \nonumber
\end{align} 
the optimization problem in (\ref{global_opt_problem_MWMSE_CSIError}) is equivalently written as  
\begin{subequations}  \label{Prob:MMSE_CSI_Error_final}
\begin{align}
\underset{\mathbb{V},\mathbb{U}, \mathbb{T}, \mathbb{M}}{ \text{min}}  \;\;\; & \sum_{i,j \in \mathbb{I}} \sum_{k \in \mathbb{F}_K} \tau_{ij}^k  \\
{\text{s.t.}} \;\;\;\;\;\;     &  {\ma{F}}_{i,j}^k \succeq 0,  \; \ma{G}_i \succeq 0 , \;\;  \forall i,j \in \mathbb{I},\;k \in \mathbb{F}_K, 
\end{align} 
\end{subequations}
where $\mathbb{M} :=\{\lambda_{ij}^k, \; \forall i,j \in \mathbb{I}, k \in \mathbb{F}_K \}$, and
\begin{align} 
\ma{G}_{i} :&= \left[\hspace{-0mm}\begin{array}{cc} P_i &\hspace{-0mm}  \tilde{\ma{v}}_i^H  \\ \tilde{\ma{v}}_i  & \hspace{-0mm}\ma{I}_{} \end{array} \hspace{-0mm}\right], \; \tilde{\ma{v}}_i  := \left\lfloor \text{vec} \left( \left( \ma{I} + K\ma{\Theta}_{\text{tx},i} \right)^{\frac{1}{2}} \ma{V}_{i}^k \right) \right\rfloor_{k\in \mathbb{F}_K} , \nonumber
\end{align}
\begin{align}
{\ma{F}}_{i,j}^k :&= \left[\begin{array}{ccc} \tau_{ij}^k - \lambda_{ij}^k & {{\ma{c}}_{ij}^k}{}^H &  \ma{0}_{1 \times M_i N_j }\\  {\ma{c}}_{ij}^k & \ma{I}_{\tilde{d}_{ij}}  &  - \zeta_{ij}^k {\ma{C}}_{ij}{}^k\tilde{\ma{D}}_{ij}^k \\  \ma{0}_{{M_i N_j \times 1 }} & - \zeta_{ij}^k {\tilde{\ma{D}}_{ij}^k}{}^H {{\ma{C}}_{ij}^k}{}^H  &   \lambda_{ij}^k \ma{I}_{M_i N_j} \end{array} \right]. \nonumber
\end{align} 
Similar to (\ref{eq:global_opt_problem_rate}), the obtained problem in (\ref{Prob:MMSE_CSI_Error_final}) is not a jointly, but a separately convex problem over $\mathbb{V}$ and $\mathbb{U}$, in each case when the other variables are fixed. In particular, the optimization over $\mathbb{V},\mathbb{T}, \mathbb{M}$ is cast as an SDP, assuming a fixed $\mathbb{U}$. Afterwards, the optimization over $\mathbb{U}, \mathbb{T}, \mathbb{M}$ is solved as an SDP, assuming a fixed $\mathbb{V}$. The described alternating steps are continued until a stable point is achieved, see Algorithm~\ref{alg:WMMSECSIError} for the detailed procedure.  
\vspace{-1mm}
\begin{algorithm}[H] 
 \small{	\begin{algorithmic}[1] 
\State{$\ell \leftarrow  {0}  ;  \;\;\;\; \text{(set iteration number to zero)}$}
\State{$\mathbb{V},\mathbb{U} \leftarrow  \text{similar~initialization~as~Algorithm~$1$}$}
\Repeat 
\State{$\ell \leftarrow  \ell + 1$}
\State{$\mathbb{V}, \mathbb{T}, \mathbb{M} \leftarrow \text{solve SDP (\ref{Prob:MMSE_CSI_Error_final}), with fixed} \; \mathbb{U}$}
\State{$\mathbb{U}, \mathbb{T}, \mathbb{M} \leftarrow \text{solve SDP (\ref{Prob:MMSE_CSI_Error_final}), with fixed} \; \mathbb{V}$}
\Until{$\text{a stable point, or maximum number of $\ell$ reached}$}

\State{\Return$\left\{ \mathbb{U},\mathbb{V} \right\}$}
  \end{algorithmic} } 
 \caption{\small{Alternating SDP (AltSDP) for worst-case MMSE design under CSI error.} } \label{alg:WMMSECSIError}
\end{algorithm}
\vspace{-3mm}
\subsection{WMMSE for sum rate maximization} \label{sec:Error_CSI_Rate}
Under the impact of CSI error, the worst-case rate maximization problem is written as 
\begin{subequations} \label{eq:model_optimization_Error_CSI_Rate}
\begin{align}
\underset{\mathbb{V}}{ \text{max}} \;\; \underset{\mathbb{C} }{ \text{min}} \;\; & \;\;\sum_{i \in \mathbb{I}} \sum_{k \in \mathbb{F}_K}  I_i^k   \\
{\text{s.t.}} \;\; &  \text{(\ref{eq:global_opt_problem_MWMSE_b})}, \;\; \ma{\Delta}_{ij}^{k} \in \mathbb{D}_{ij}^{k}, \;\;  \forall i,j \in \mathbb{I}, \;\;  k \in \mathbb{F}_K.  \label{eq:global_opt_problem_1_normbounded_contraint_Error_CSI_Rate}
\end{align} 
\end{subequations}
Via the application of Lemma~\ref{lemma_logdetE}, and (\ref{wmmse_wmmseEquivalent}) the rate maximization problem is equivalently written as 
\begin{subequations}  \label{eq:global_opt_problem_2_Error_CSI_Rate}
\begin{align}
\underset{\mathbb{V}}{ \text{max}} \;\; \underset{\mathbb{C}}{ \text{min}}  \;\; \underset{\mathbb{U}, \mathbb{W}}{ \text{max}} \;\;  & \sum_{i \in \mathbb{I}}   \sum_{k \in \mathbb{F}_K}  \bigg( \text{log} \left|\ma{W}_i^k{\ma{W}_i^k}^H\right|  \nonumber \\ & + d_i - \text{tr} \left({\ma{W}_i^k}^H \ma{E}_i^k\ma{W}_i^k\right) \bigg) \\
{\text{s.t.}} \;\;     & \text{(\ref{eq:global_opt_problem_1_normbounded_contraint_Error_CSI_Rate})}, \label{eq:global_opt_problem_2_constraints}
\end{align} 
\end{subequations}
where $\mathbb{W}:= \{\ma{W}_{i}^k, \; \forall i \in \mathbb{I}, \; \forall k \in \mathbb{F}_K\}$. The above problem is not tractable in the current form, due to the inner min-max structure. Following the max-min exchange introduced in \cite[Section~III]{JPKR:11}, and undertaking similar steps as in (\ref{quadratic_error_representation_final})-(\ref{CSI_Error_ChoosingMatrices_Lemma_Peterson}) the problem (\ref{eq:global_opt_problem_2_Error_CSI_Rate}) is turned into 
\begin{subequations}  \label{Prob:RATE_CSI_Error_final}
\begin{align}
\underset{\mathbb{V},\mathbb{U}, \mathbb{W}, \mathbb{T}, \mathbb{M} }{ \text{max}}  \;\;\; & \sum_{i \in \mathbb{I}} \sum_{k \in \mathbb{F}_K}  \bigg( 2 \text{log}\left|\ma{W}_i^k\right|  + d_i - \sum_{j \in \mathbb{I}} \tau_{ij}^k \bigg)  \\
{\text{s.t.}} \;\; \;\; \;\; \;\;     &  {\ma{F}}_{i,j}^k \succeq 0,  \; \ma{G}_i \succeq 0 , \;\;  \forall i,j \in \mathbb{I},\;k \in \mathbb{F}_K, 
\end{align} 
\end{subequations}
where ${\ma{F}}_{i,j}^k, \ma{G}_i$ are defined in (\ref{Prob:MMSE_CSI_Error_final}). It is observable that the transformed problem holds a separately, but not a jointly, convex structure over the optimization variable sets. In particular, the optimization over $\mathbb{V},\mathbb{T}, \mathbb{M}$ and $\mathbb{U},\mathbb{T}, \mathbb{M}$ are cast as SDP in each case when other variables are fixed. Moreover, the optimization over $\mathbb{W}$ can be efficiently implemented using the MAX-DET algorithm \cite{vandenberghe1998determinant}, see Algorithm~\ref{alg:WMMSECSIError_Rate}. Similar to Algorithm~\ref{alg:WMMSECSIError}, due to the monotonic increase of the objective in each optimization iteration the algorithm convergences to a stationary point. See \cite[Section~III]{JPKR:11} for arguments regarding convergence and optimization steps for a problem with a similar variable separation.   

%

\begin{algorithm}[H] 
 \small{	\begin{algorithmic}[1] 
\State{$\text{Algorithm~1, Steps 1-3} \;\; \text{(initialization)}$}


\Repeat 
\State{$\mathbb{W}, \mathbb{T}, \mathbb{M} \leftarrow \text{solve MAT-DET (\ref{Prob:MMSE_CSI_Error_final}), with fixed} \; \mathbb{V},\mathbb{U}$}
\State{$\text{Algorithm~3, Steps 4-6}$}
\Until{$\text{a stable point, or maximum number of $\ell$ reached}$}

\State{\Return$\left\{ \mathbb{U},\mathbb{V} \right\}$}
  \end{algorithmic} } 
 \caption{\small{AltSDP-WMMSE algorithm for worst-case rate maximization under CSI error} } \label{alg:WMMSECSIError_Rate}
\end{algorithm}  \vspace{-3mm}

\vspace{-4mm}

\subsection{Worst case CSI error} \label{WC_CSI}
It is beneficial to obtain the least favorable CSI error matrices, as they provide guidelines for the future channel estimation strategies. For instance, this helps us to choose a channel training sequence that reduces the radius of the CSI error feasible regions in the most destructive directions. Moreover, such knowledge is a necessary step for cutting-set-based methods \cite{mutapcic2009cutting} which aim to reduce the design complexity by iteratively identifying the most destructive error matrices and explicitly incorporating them into the future design steps. In the current setup, the worst-case channel error matrices are identified by maximizing the weighted MSE objective in (\ref{eq:global_opt_problem_MWMSE_CSIError}) within their defined feasible region. This is expressed as
\begin{subequations}  \label{wmmse_}
\begin{align}
\underset{{{\mathbb{C}}}}{ \text{max}} \;\;  & \sum_{i \in \mathbb{I}} \sum_{k\in\mathbb{F}_K}\text{tr}\left({\ma{W}_{i}^k}{}^H \ma{E}_i^k \ma{W}_{i}^k\right), \;\; \\
{\text{s.t.}} \;\;     &  \left\| \ma{D}_{ij}^k  \ma{\Delta}_{ij}^k \right\|_F \leq \zeta_{ij}^k, \;\; \forall i,j \in \mathbb{I},\; k \in \mathbb{F}_K.
\end{align} 
\end{subequations}
Due to the uncoupled nature of the error feasible set, and the value of the objective function over $\ma{\Delta}_{ij}^k$, following (\ref{quadratic_error_representation_final}), the above problem is decomposed as 
\begin{subequations} \label{find_worst_delta_3}
\begin{align} 
\underset{ {\ma{b}}_{ij}^k }{ \text{min}} \;\;  & - \left\| \ma{C}_{ij}^k \tilde{\ma{D}}_{ij}^k {\ma{b}}_{ij}^k \right\|_2^2 - 2 \text{Re}\left\{ {{\ma{b}}_{ij}^k}{}^H   {{\tilde{\ma{D}}_{ij}^k}}{}^H {\ma{C}_{ij}^k}{}^H  {\ma{c}}_{ij}^k \right\} -  {{\ma{c}}_{ij}^k}{}^H {\ma{c}}_{ij}^k \label{wwmse_worstcaseerror_nonconvexquadraticobjective}\\ {\text{ s.t.}} \;\;      & {{\ma{b}}_{ij}^k}{}^H {\ma{b}}_{ij}^k  \leq {\zeta_{ij}^k}{}^2, 
\end{align} 
\end{subequations}
where $\text{Re}\{\cdot\}$ represents the real part of a complex value. Note that the objective in (\ref{wwmse_worstcaseerror_nonconvexquadraticobjective}) is a non convex function and can not be minimized using the usual numerical solvers in the current form. Following the zero duality gap results for the non-convex quadratic problems \cite{zheng2012zero}, we focus on the dual function of (\ref{find_worst_delta_3}). The corresponding Lagrangian function to (\ref{find_worst_delta_3}) is constructed as 
\begin{align}  \label{find_worst_delta_Lagrangian}
& \mathcal{L}\left({\ma{b}}_{ij}^k , \rho_{ij}^k \right)  =  \nonumber  \\  &  {{\ma{b}}_{ij}^k}{}^H  \ma{A}_{ij}^k {\ma{b}}_{ij}^k  - 2 \text{Re}\left\{ {{\ma{b}}_{ij}^k}{}^H   {\tilde{\ma{D}}_{ij}^k}{}^H {\ma{C}_{ij}^k}{}^H  {\ma{c}}_{ij}^k \right\} -  {{\ma{c}}_{ij}^k}{}^H {\ma{c}}_{ij}^k - \rho_{ij}^k  {\zeta_{ij}^k}{}^2, 
\end{align} 
where $\rho_{ij}^k$ is the dual variable and 
\begin{align} 
\ma{A}_{ij}^k := \rho_{ij}^k \ma{I}_{N_jM_i} - {\tilde{\ma{D}}_{ij}^k}{}^H{{\ma{C}}_{ij}^k}{}^H{\ma{C}}_{ij}{}^k\tilde{\ma{D}}_{ij}^k.
\end{align} 
Consequently, the value of the dual function is obtained as 
\begin{align}
& \ma{g} \left( \rho_{ij}^k \right) = \nonumber \\ & - {{{\ma{c}}}_{ij}^k}{}^H  \ma{C}_{ij}^k  \tilde{\ma{D}}_{ij}^k \left(\ma{A}_{ij}^k\right)^{-1} {\tilde{\ma{D}}_{ij}^k}{}^H {\ma{C}_{ij}^k}{}^H  {{\ma{c}}}_{ij}^k  - {{\ma{c}}_{ij}^k}{}^H{\ma{c}}_{ij}^k - \rho_{ij}^k {\zeta_{ij}^k}{}^2, \nonumber
\end{align}
if $ \ma{A}_{ij}^k \succeq 0$, and ${\tilde{\ma{D}}_{ij}^k}{}^H {\ma{C}_{ij}^k}{}^H  {{\ma{c}}}_{ij}^k \in \mathcal{R}\{\ma{A}_{ij}^k\}$, and otherwise is unbounded from below\footnote{If one of the aforementioned conditions is not satisfied, an infinitely large value of $\ma{b}_{ij}$ can be chosen in the negative direction of $\ma{A}_{ij}$, if $\ma{A}_{ij}^k$ is not positive semi-definite, or in the direction ${\tilde{\ma{D}}_{ij}^k}{}^H {\ma{C}_{ij}^k}{}^H  {\ma{c}}_{ij}^k$ within the null-space of $\ma{A}_{ij}^k$.}. By applying the Schur complement lemma, the maximization of the dual function is written using the epigraph form as
\begin{subequations} \label{eq:dual_channelerrormatrices}
\begin{align}
\underset{ {\rho}_{ij}^k \geq 0 , \; \phi_{ij}^k} { \text{max} } \;\;  & - \phi_{ij}^k \\  
{\rm s.t.} \;\; & \left[\begin{array}{cc} \phi_{ij}^k - {{\ma{c}}_{ij}^k}{}^H {\ma{c}}_{ij}^k  - {\rho}_{ij}^k {\zeta_{ij}^k}{}^2   & {{{\ma{c}}}_{ij}^k}{}^H  \ma{C}_{ij}^k  \tilde{\ma{D}}_{ij}^k  \\ {\tilde{\ma{D}}_{ij}^k}{}^H {\ma{C}_{ij}^k}{}^H  {{\ma{c}}}_{ij}^k   & \ma{A}_{ij}^k  \end{array} \right] \succeq 0,  \label{eq:dual_channelerrormatrices_schur_semidefiniteconstraint}
\end{align} 
\end{subequations}
where $\phi_{ij}^k \in \real $ is an auxiliary variable\footnote{Note that the semi-definite presentation in (\ref{eq:dual_channelerrormatrices_schur_semidefiniteconstraint}) automatically satisfies  $ \ma{A}_{ij}^k \succeq 0$, and ${\tilde{\ma{D}}_{ij}^k}{}^H {\ma{C}_{ij}^k}{}^H  {{\ma{c}}}_{ij}^k \in \mathcal{R}\{\ma{A}_{ij}^k\}$.}. By plugging the obtained dual variable ${\rho}_{ij}^k$ into (\ref{find_worst_delta_Lagrangian}), and considering the fact that $- {\tilde{\ma{D}}_{ij}^k}{}^H {\ma{C}_{ij}^k}{}^H \ma{C}_{ij}^k \tilde{\ma{D}}_{ij}^k + {\rho_{ij}^k}{}^\star \ma{I}_{N_jM_i} \succeq 0$ as a result of (\ref{eq:dual_channelerrormatrices}), the optimal value of $\ma{b}_{ij}^k$ is obtained from (\ref{find_worst_delta_Lagrangian}) as 
\begin{align}
{\ma{b}_{ij}^k}{}^\star = \left( - {\tilde{\ma{D}}_{ij}^k}{}^H {\ma{C}_{ij}^k}{}^H \ma{C}_{ij}^k \tilde{\ma{D}}_{ij}^k + {\rho_{ij}^k}{}^\star \ma{I}_{N_jM_i} \right)^{-1} {\tilde{\ma{D}}_{ij}^k}{}^H {\ma{C}_{ij}^k}{}^H  {{\ma{c}}}_{ij}^k, \nonumber
\end{align} 
where $(\cdot)^\star$ represents the optimality and the worst case $\ma{\Delta}_{ij}^k$ is consequently calculated via $\text{vec}(\ma{\Delta}_{ij}^k) = \tilde{\ma{D}}_{ij}^k {\ma{b}_{ij}^k}{}^\star$.

\subsection{Computational complexity} \label{alg_complexity}
The proposed designs in Section~\ref{sec:WMMSE} and \ref{sec:WMMSE_CSI_Error} are based on the alternative design of the optimization variables. Furthermore, it is observed that the consideration of non-linear hardware distortions, leading to inter-carrier leakage, as well as the impact of CSI error, result in a higher problem dimension and thereby complicate the structure of the resulting optimization problem. In this part, we analyze the arithmetic complexity associated with the Algorithm~\ref{sec:WMMSE_CSI_Error}. Note that Algorithm~\ref{sec:WMMSE_CSI_Error} is considered as a general framework, containing Algorithm~\ref{sec:WMMSE} as a special case, since it takes into account the impacts of hardware distortion jointly with CSI error. \par
The optimization over $\mathbb{V}, \mathbb{U}$ are separately cast as SDP. A general SDP problem is defined as 
\begin{align}
\underset{\ma{z}}{\text{min}} \;\; \ma{p}^T \ma{z}, \;\; {\text{s.t.}} \;\;  \ma{z}\in \real^n, \; \ma{Y}_0 + \sum_{i=1}^n z_i \ma{Y}_i \succeq 0, \; \|\ma{z}\|_2 \leq q,  \nonumber 
\end{align}
where the fixed matrices $\ma{Y}_i$ are symmetric block-diagonal, with $M$ diagonal blocks of the sizes $l_m \times l_m,\; m \in \mathbb{F}_M$, and define the specific problem structure, see \cite[Subsection~4.6.3]{ben2001lectures}. The arithmetic complexity of obtaining an $\epsilon$-solution to the defined problem, i.e., the convergence to the $\epsilon$-distance vicinity of the optimum is upper-bounded by 
\begin{align}
\mathcal{O}(1) \left(1 + \sum_{m=1}^{M} l_m  \right)^{\frac{1}{2}} \left( n^3 + n^2 \sum_{m=1}^{M} l_m^2 + n \sum_{m=1}^M l_m^3\right) \text{digit}\left( \epsilon\right), \nonumber
\end{align}
where \rev{$\mathcal{O}(1)$ is a positive constant and invariant to the problem dimensions~\cite{ben2001lectures}}, and $\text{digit}(\epsilon)$ is obtained from \cite[Subsection~4.1.2]{ben2001lectures} and affected by the required solution precision. The required computation of each step is hence determined by size of the variable space and the corresponding block diagonal matrix structure, which is obtained in the following:   

\subsubsection{Optimization over $\mathbb{V},\mathbb{T},\mathbb{M}$}
The size of the variable space is given as $n= 2K \left( 4 + \sum_{i\in\mathbb{I}} d_iN_i  \right)$. Moreover, the block sizes are calculated as $l_m = 2 + 2 K d_i N_i, \; \forall i \in \mathbb{I}$, corresponding to the semi-definite constraint on $\ma{G}_i$, and as $l_m = 2+ 2\tilde{d}_{ij} + 2M_iN_j, \; \forall i,j \in \mathbb{I},k\in\mathbb{F}_K$, corresponding to the semidefinite constraint on $\ma{F}_{i,j}^k$ from (\ref{Prob:MMSE_CSI_Error_final}). The overall number of the blocks is calculated as $M=2+4K$.

\subsubsection{Optimization over $\mathbb{U},\mathbb{T},\mathbb{M}$}
The size of the variable space is given as $n= 2K \left( 4 + \sum_{i\in\mathbb{I}} d_i M_i  \right)$. The block sizes are calculated as $l_m = 2+ 2\tilde{d}_{ij} + 2M_iN_j, \; \forall i,j \in \mathbb{I},k\in\mathbb{F}_K$, corresponding to the semidefinite constraint on $\ma{F}_{i,j}^k$ from (\ref{Prob:MMSE_CSI_Error_final}). The overall number of the blocks is calculated as $M=4K$.
%
%
\subsubsection{Remarks}
The above analysis intends to show how the bounds on computational complexity are related to different dimensions in the problem structure. Nevertheless, the actual computational load may vary in practice, due to the structure simplifications and depending on the used numerical solver. Furthermore, the overall algorithm complexity also depends on the number of optimization iterations required for convergence. See Subsection~\ref{sec_AlgorithmAnalysis} for a study on the convergence behavior, as well as a numerical evaluation of the algorithm computational complexity.

%



\vspace{-2mm}
\section {Simulation Results}\label{sec:simulations}
 
\begin{figure}[!t] 
    \begin{center}
        \includegraphics[angle=0,width=\MainFigureSizes]{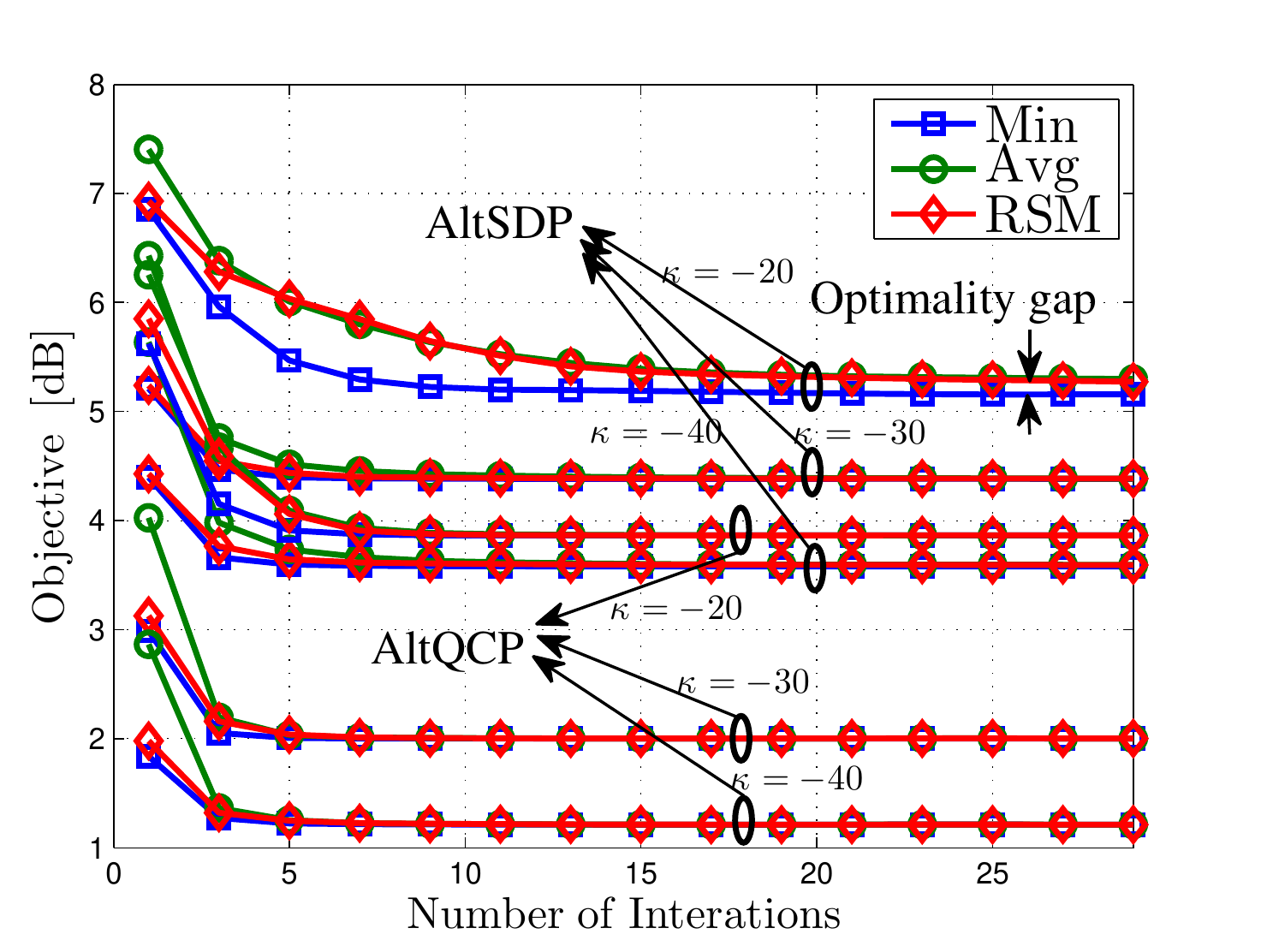}
    \caption{Average convergence behavior for AltQCP and AltSDP algorithms. AltQCP converges with fewer steps, and leads to a smaller optimality gap compared to AltSDP. Both algorithms converge in $10$-$30$ iterations.} \label{fig:init}
    \end{center} \vspace{-0mm} 
\end{figure}
\begin{figure}[t]  
\subfigure[CPU Time vs $M$]{\includegraphics[width = 0.48\columnwidth]{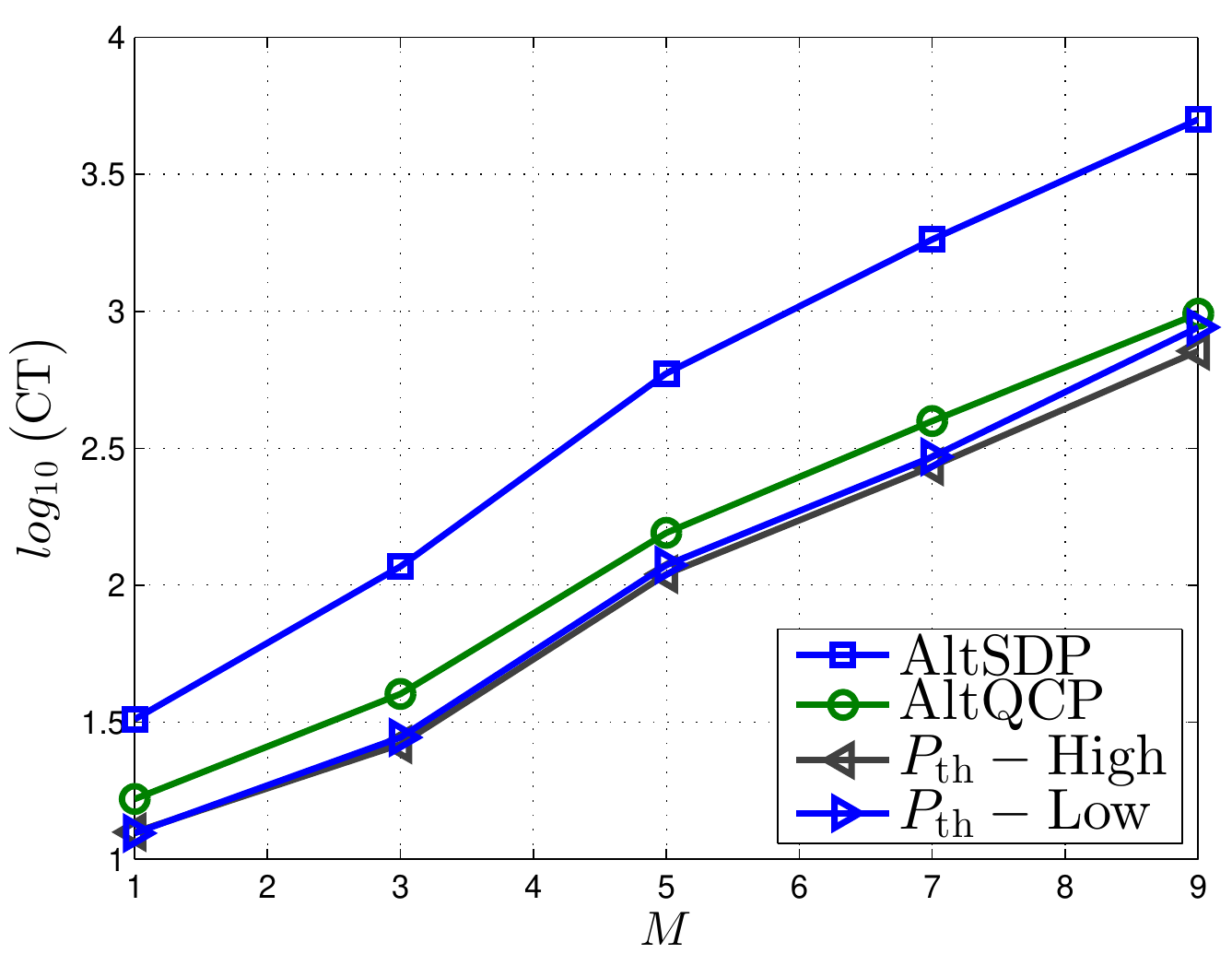}} \label{fig_conv2}
\subfigure[CPU Time vs $K$]{\includegraphics[width = 0.48\columnwidth]{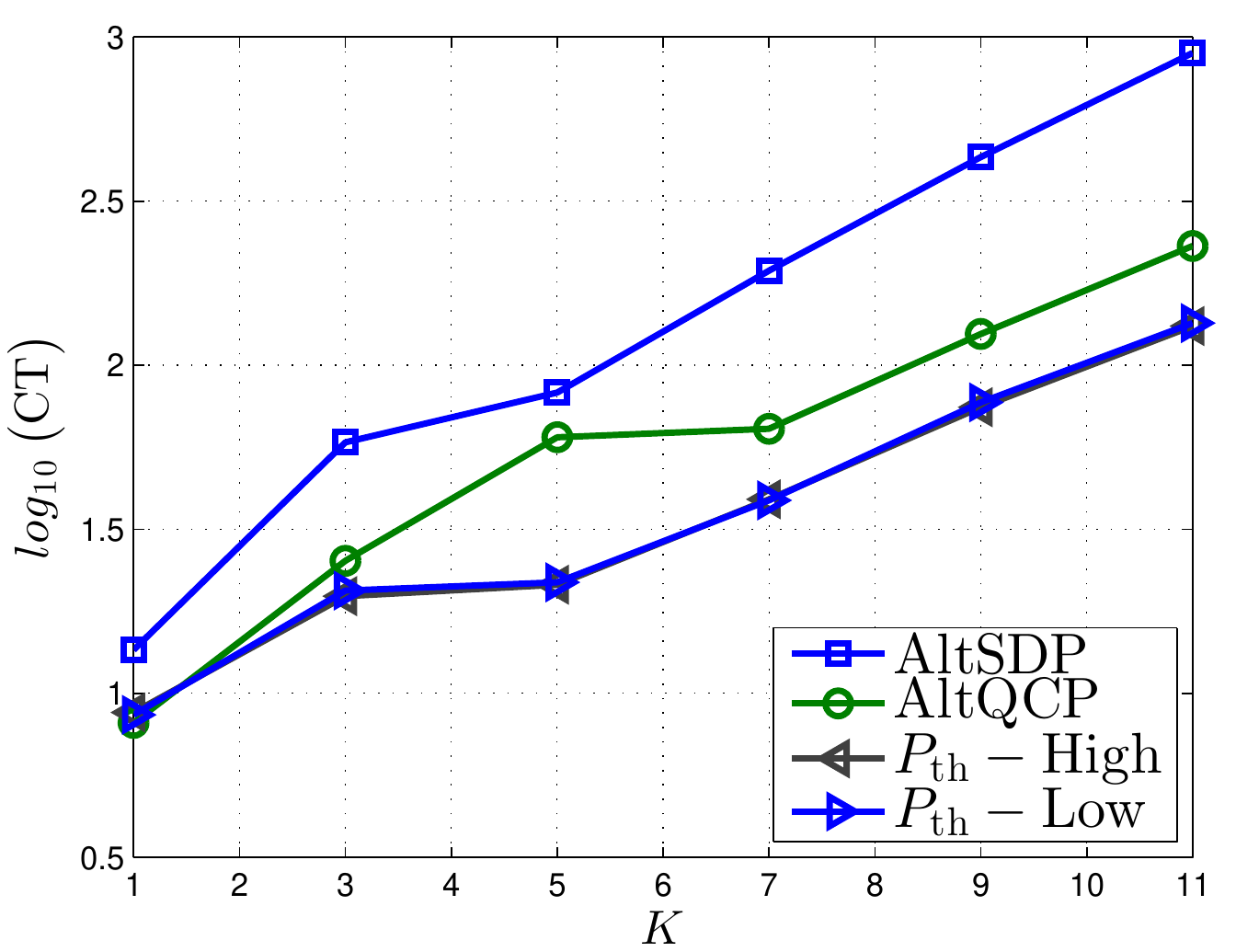}} \label{fig_conv3}
\caption{Comparison of the algorithm computational complexities, in terms of the required CPU time (CT), for different system dimensions, i.e., different $K$ and $M$. \rev{$\kappa$ represents the hardware inaccuracy, i.e., $\ma{\Theta}_{\text{rx},i} = \kappa \ma{I}_{M_{i}}, \ma{\Theta}_{\text{tx},i} = \kappa \ma{I}_{N_{i}}$}.} 
\label{fig:Complexity}
\vspace{-7pt}
\end{figure} 

\begin{figure*}[!h]  
\hspace{0.0cm} \subfigure[MSE vs. hardware inaccuracy]{\includegraphics[width = 0.666\columnwidth]{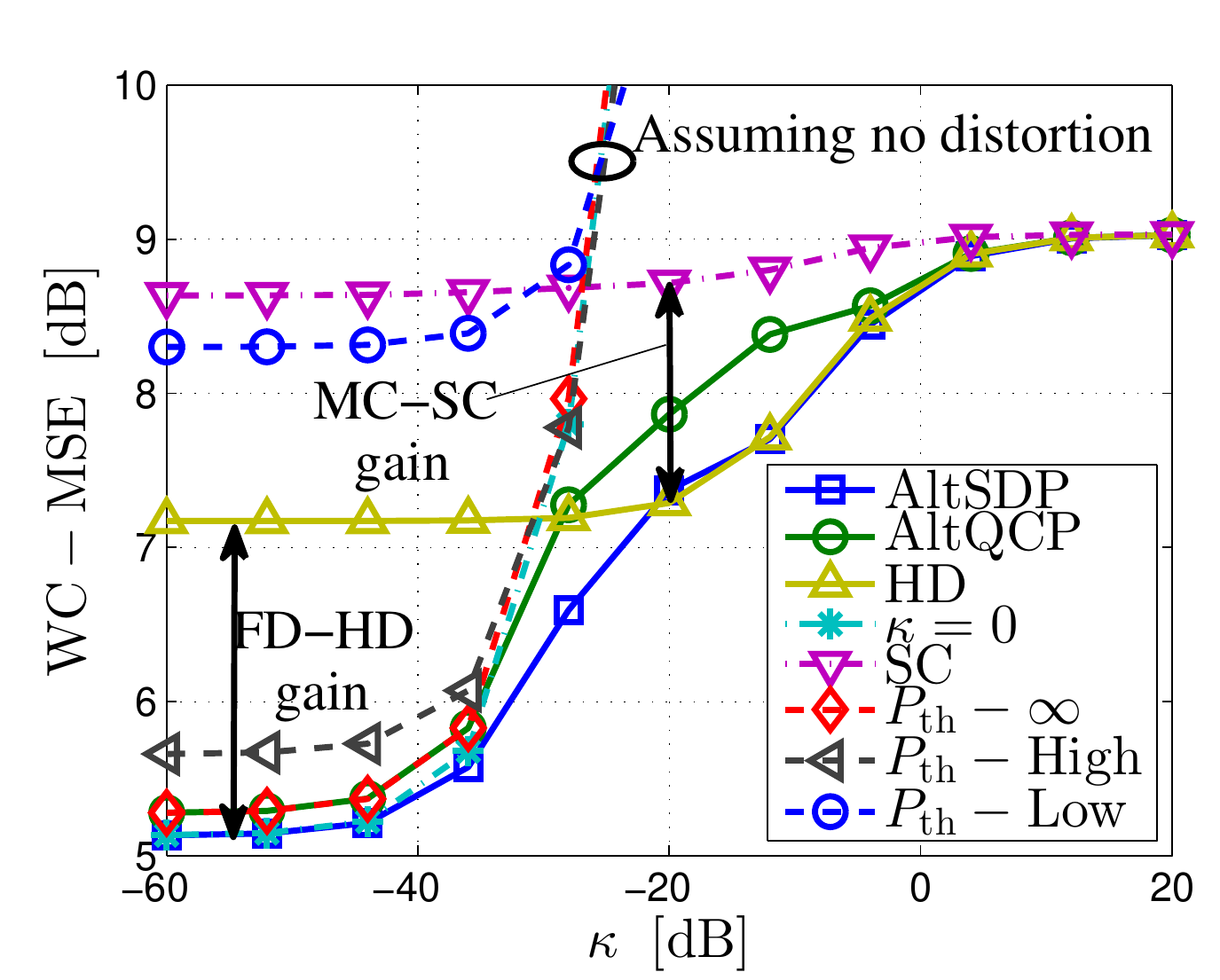}} \label{fig_sr_kappa}
\subfigure[MSE vs. CSI error ]{\includegraphics[ width = 0.666\columnwidth]{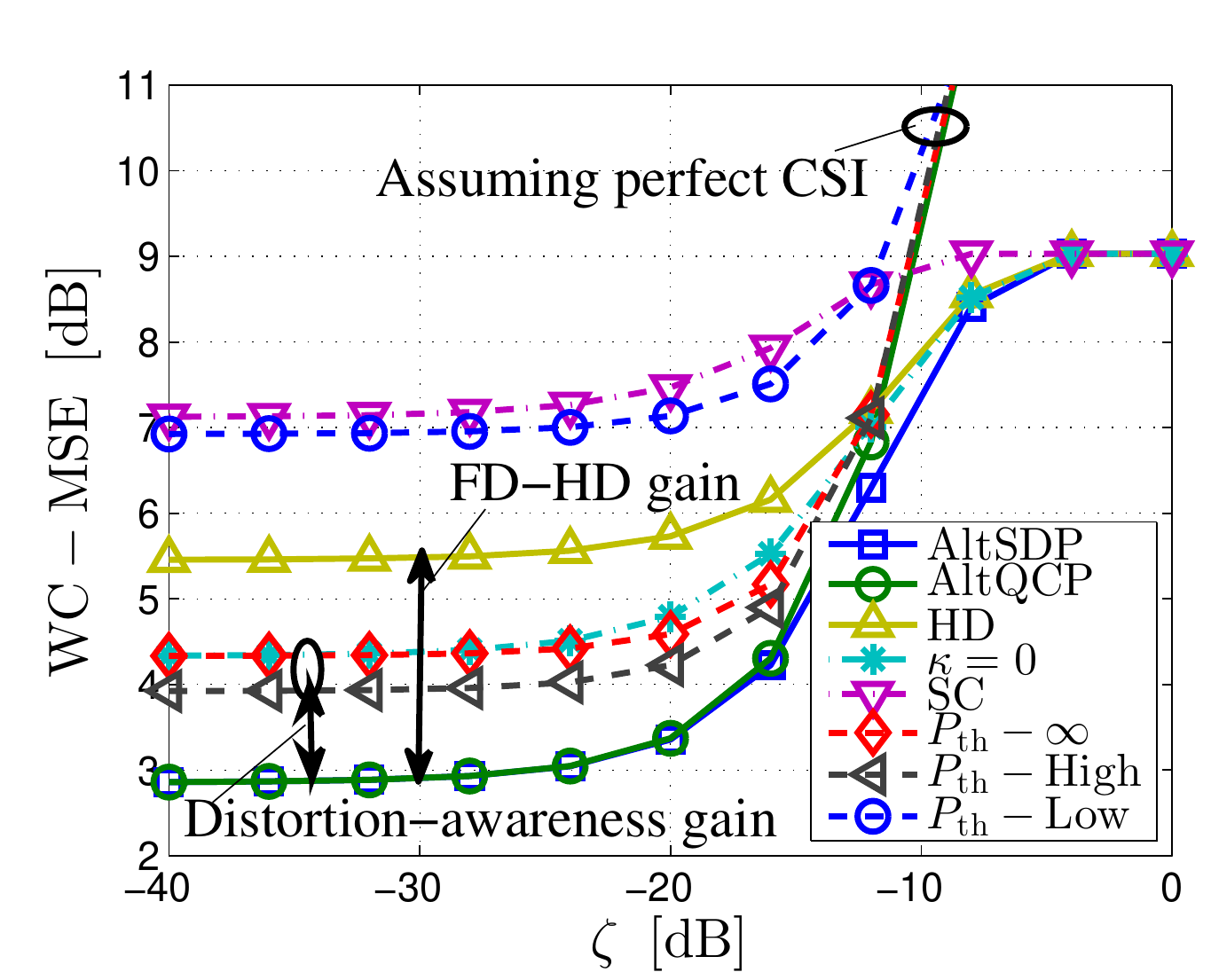}} \label{fig_sr_noise}
\subfigure[MSE vs. noise variance]{\includegraphics[width = 0.666\columnwidth]{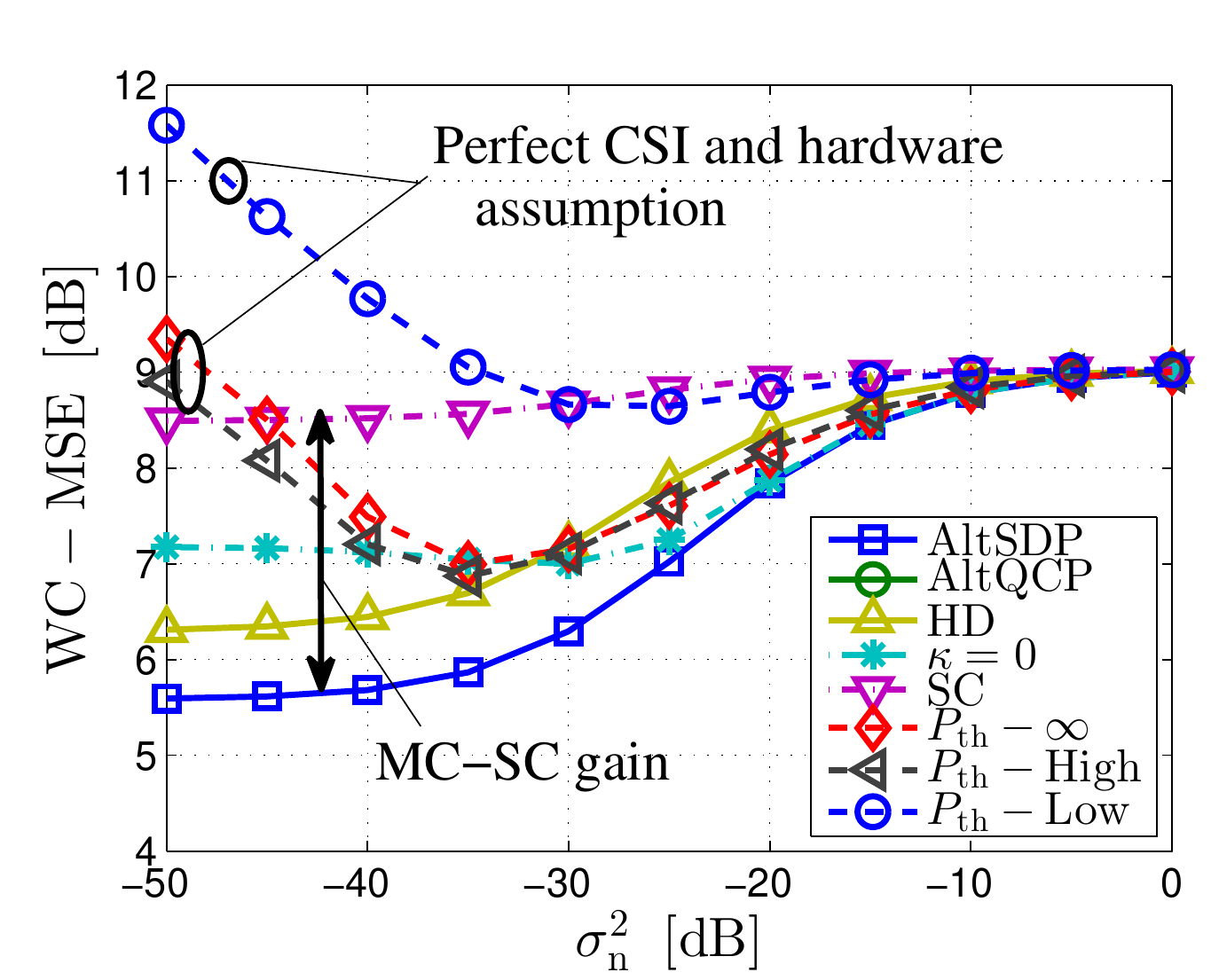}} \label{fig_sr_power}
\vspace{-3mm}\subfigure[MSE vs. comm. channel strength]{\includegraphics[width = 0.666\columnwidth]{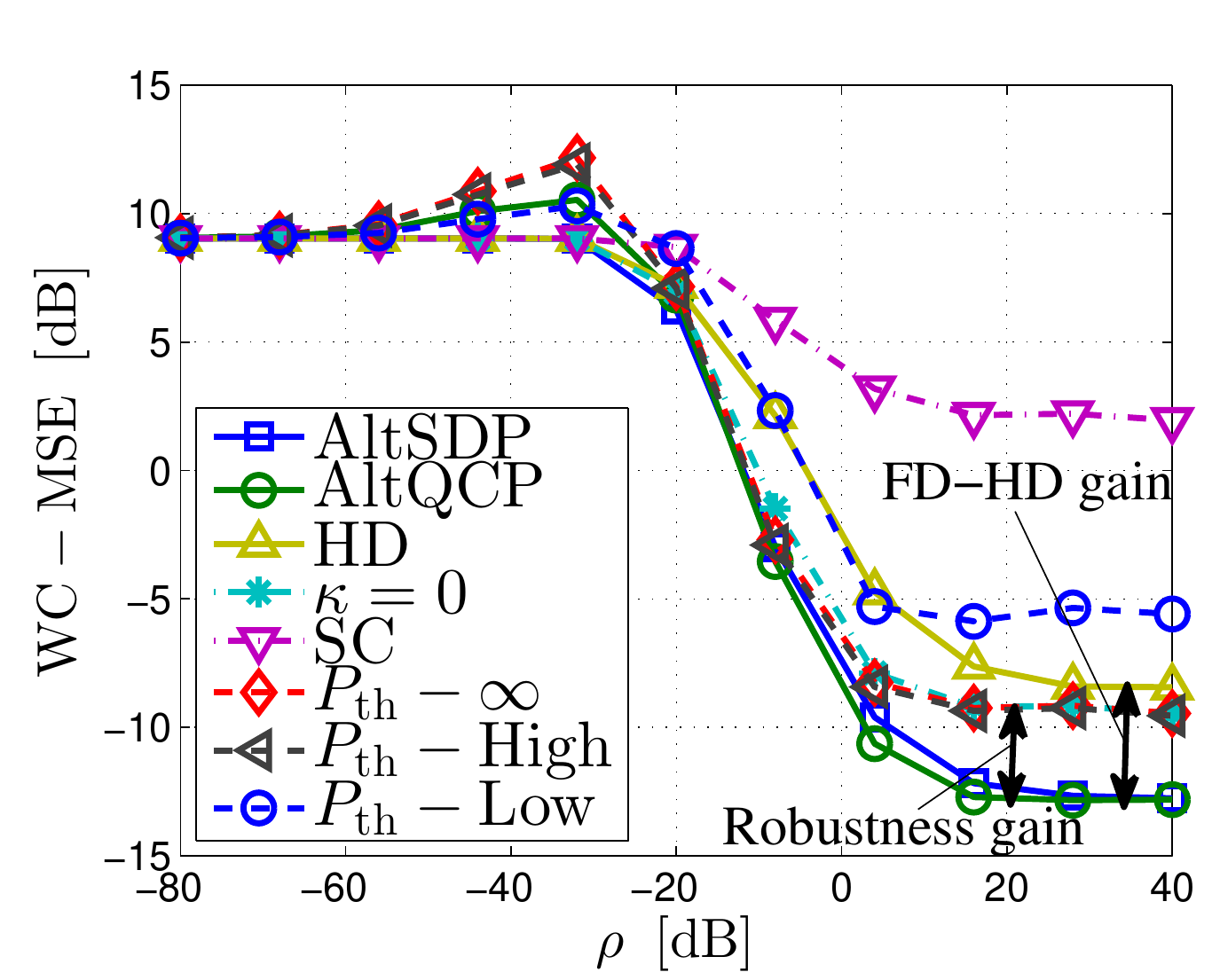}} \label{fig_sr_kappa}
\subfigure[SMSE vs. number of subcarriers]{\includegraphics[width = 0.666\columnwidth]{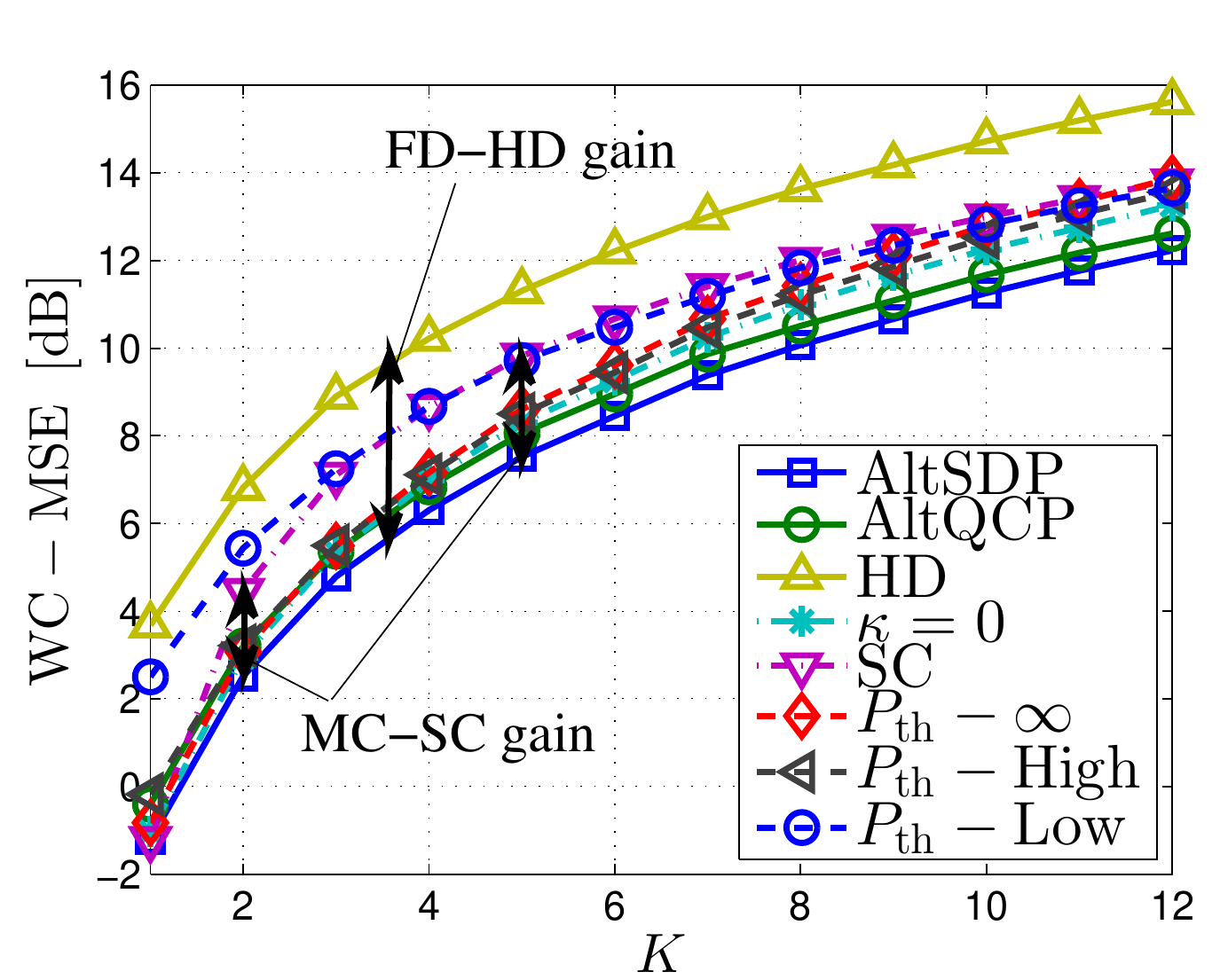}} \label{fig_sr_noise}
\subfigure[MSE vs. number of subcarriers]{\includegraphics[width = 0.666\columnwidth]{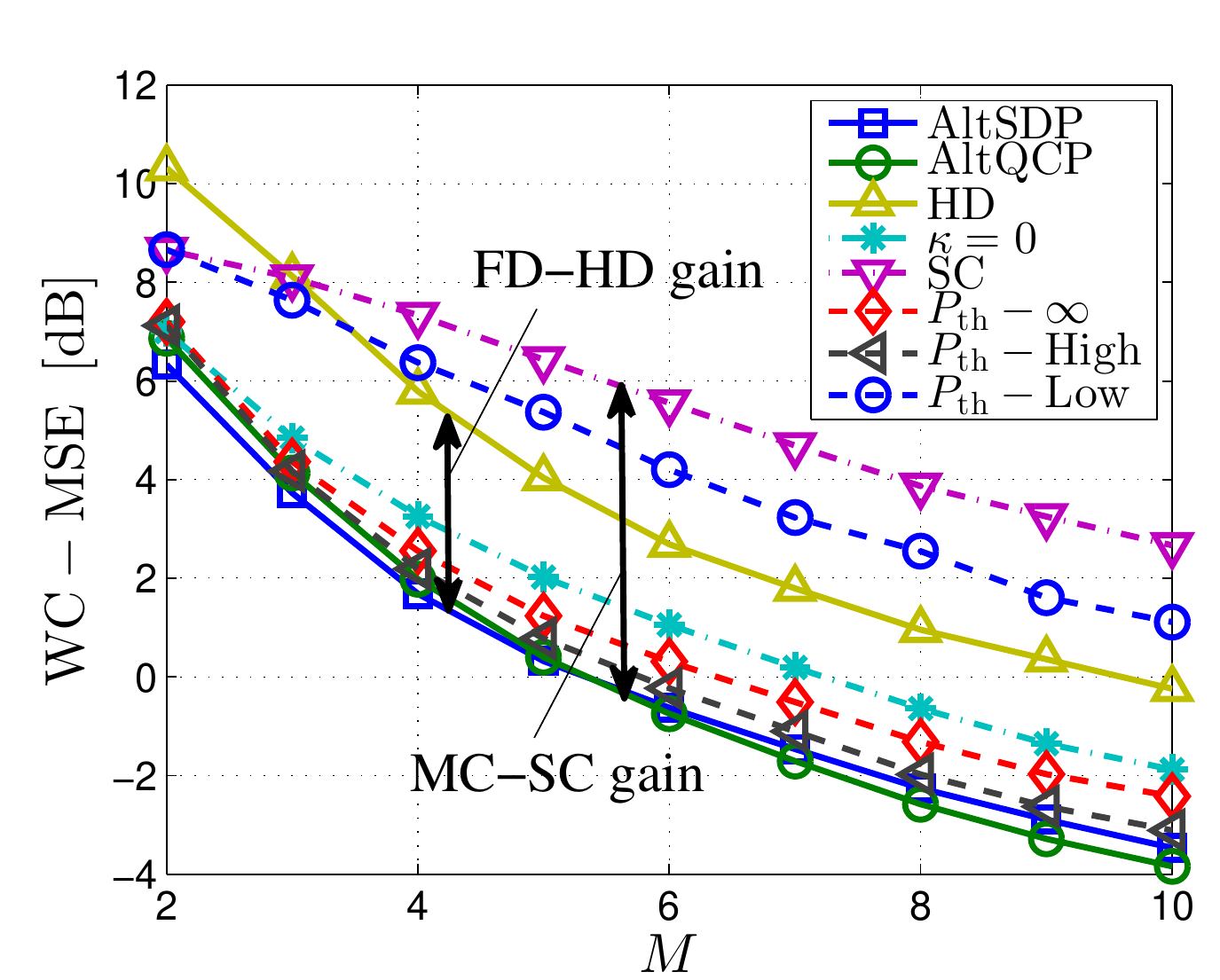}} 
\vspace{-3mm}\subfigure[MSE vs. distribution of quantization]{\includegraphics[width = 0.666\columnwidth]{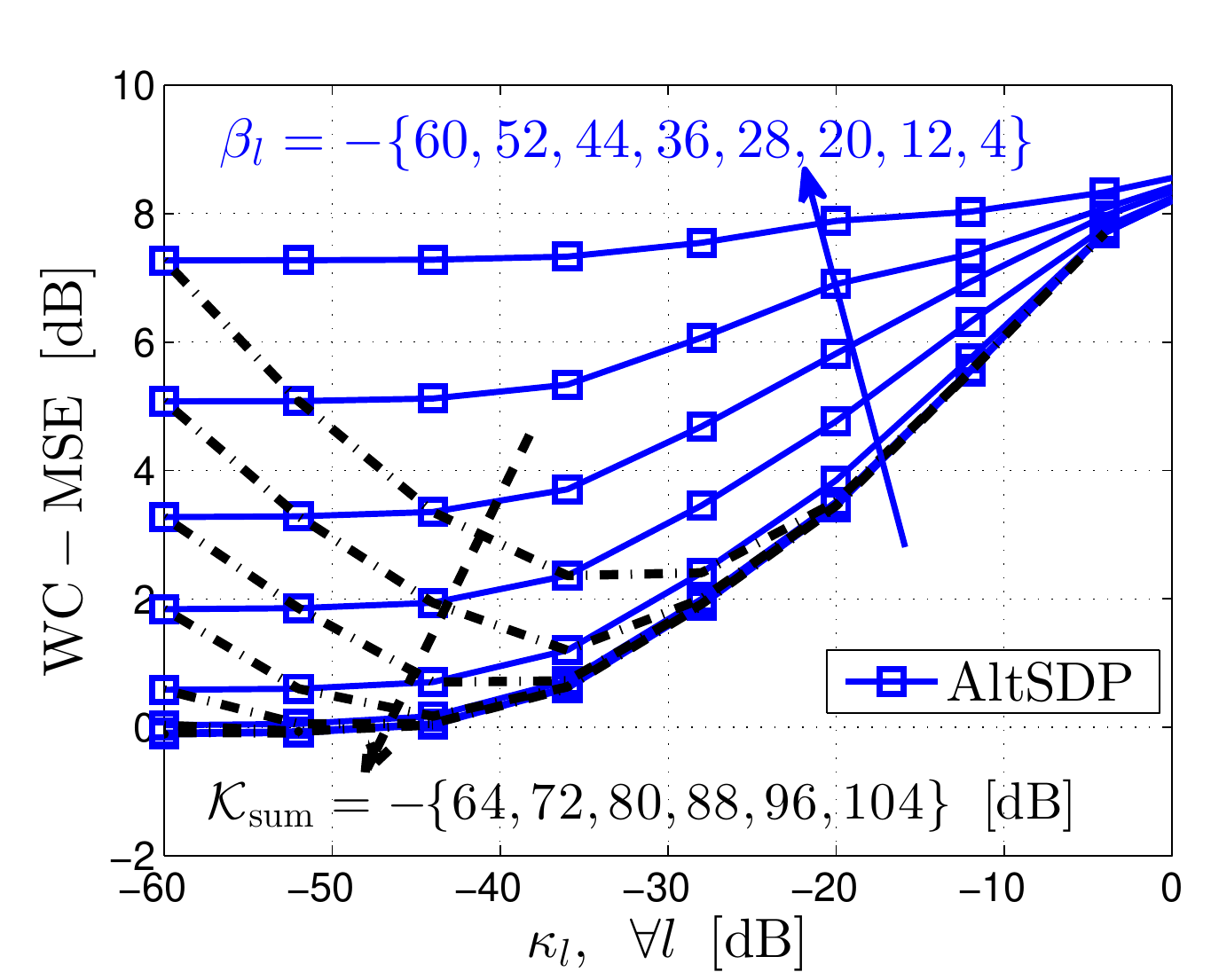}} 
\subfigure[MSE vs. distribution of antenna resources]{\includegraphics[width = 0.666\columnwidth]{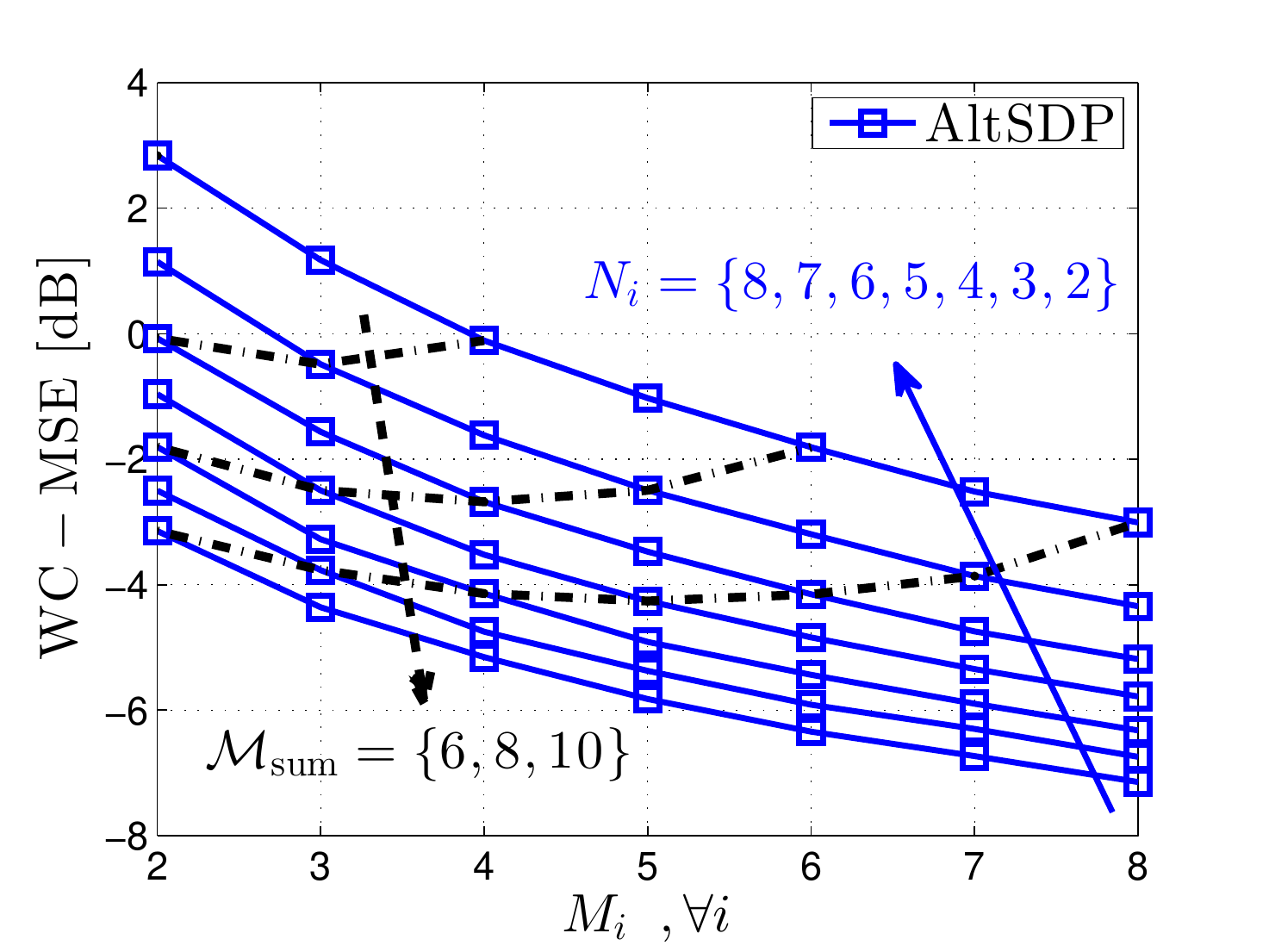}} 
\subfigure[Sum rate vs. hardware inaccuracy]{\includegraphics[width = 0.666\columnwidth]{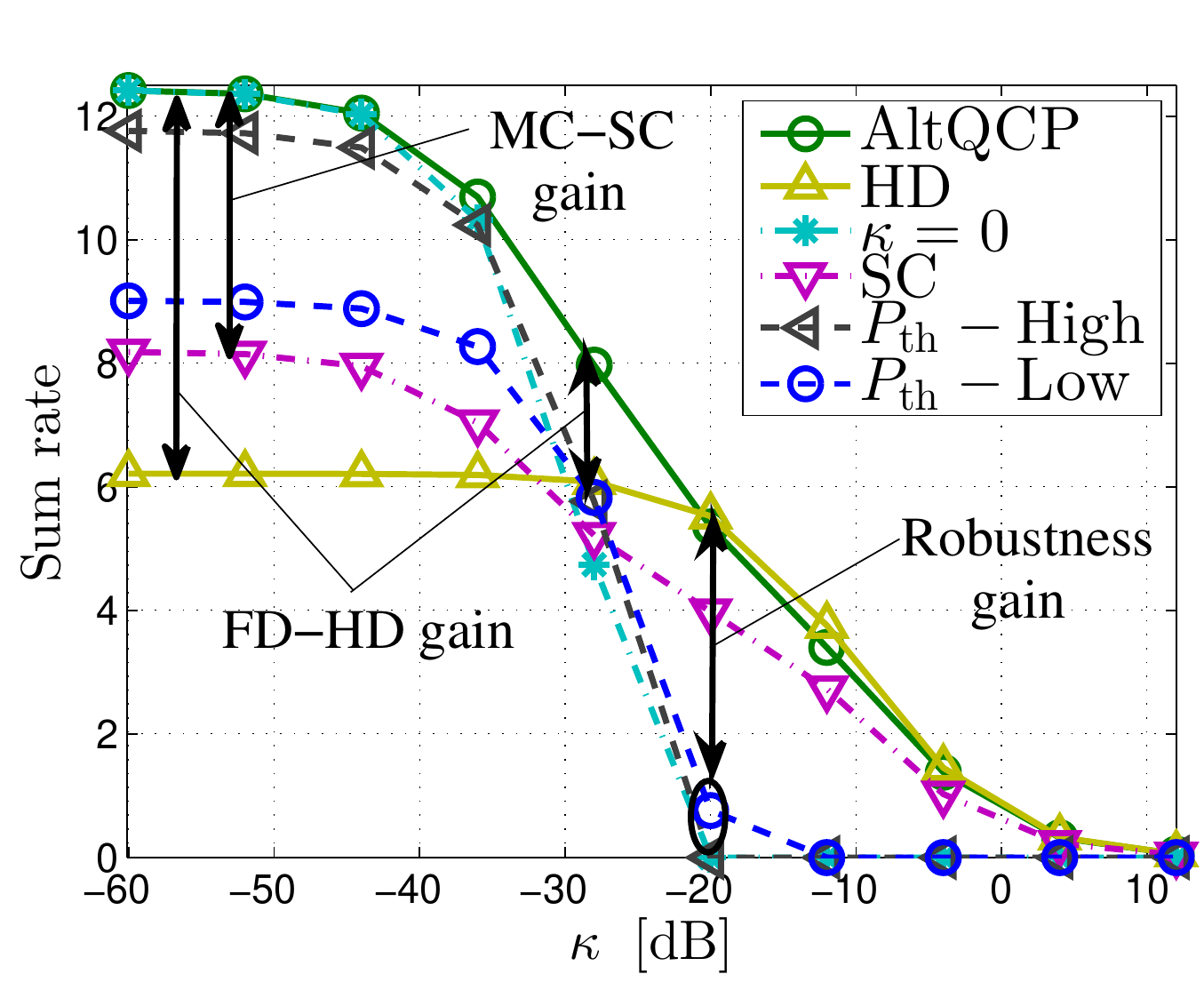}} \label{fig_sr_kappa}
\vspace{-3mm}\subfigure[Sum rate vs. noise variance]{\includegraphics[width = 0.666\columnwidth]{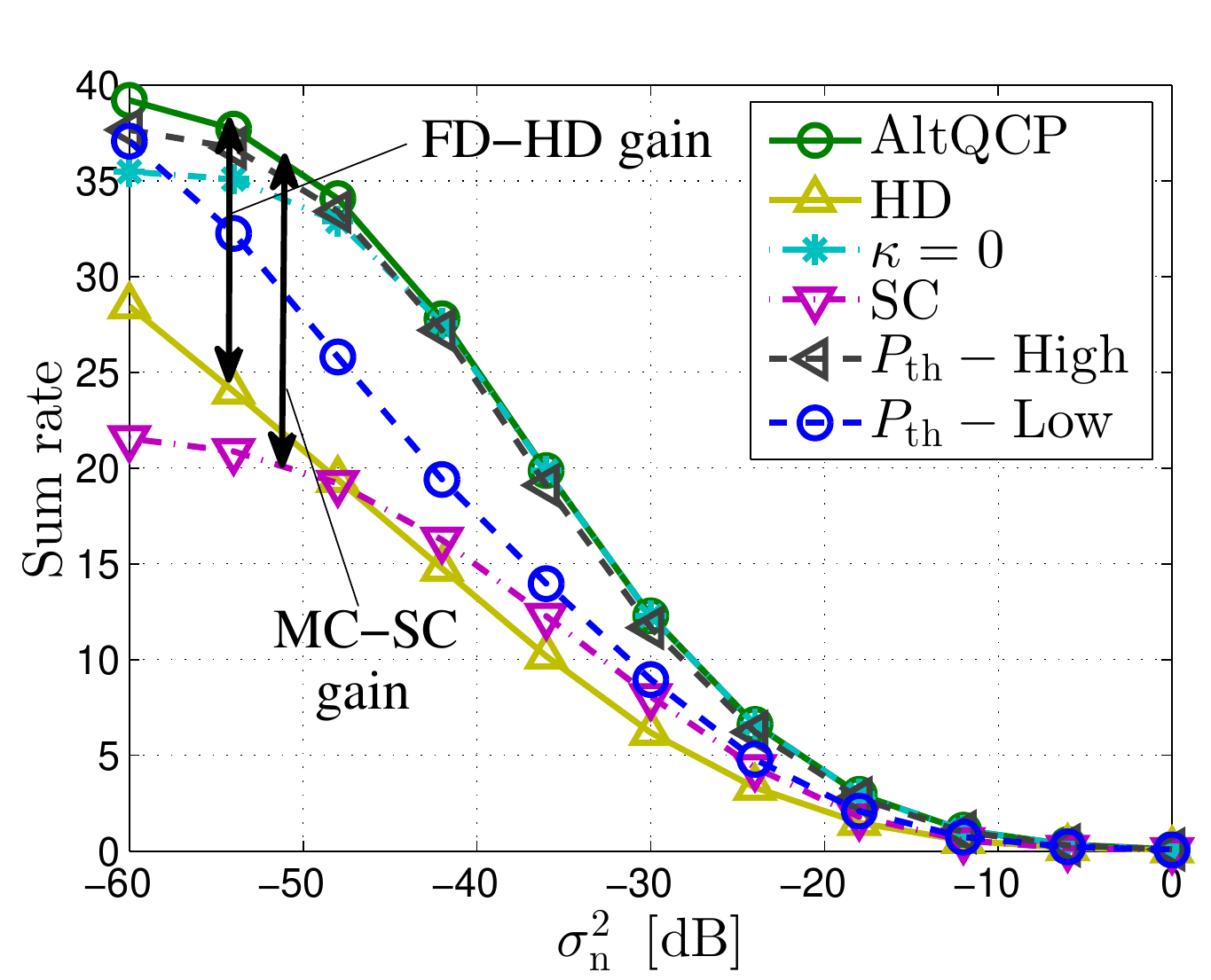}} \label{fig_sr_noise}
\subfigure[Sum rate vs. maximum transmit power]{\includegraphics[width = 0.666\columnwidth]{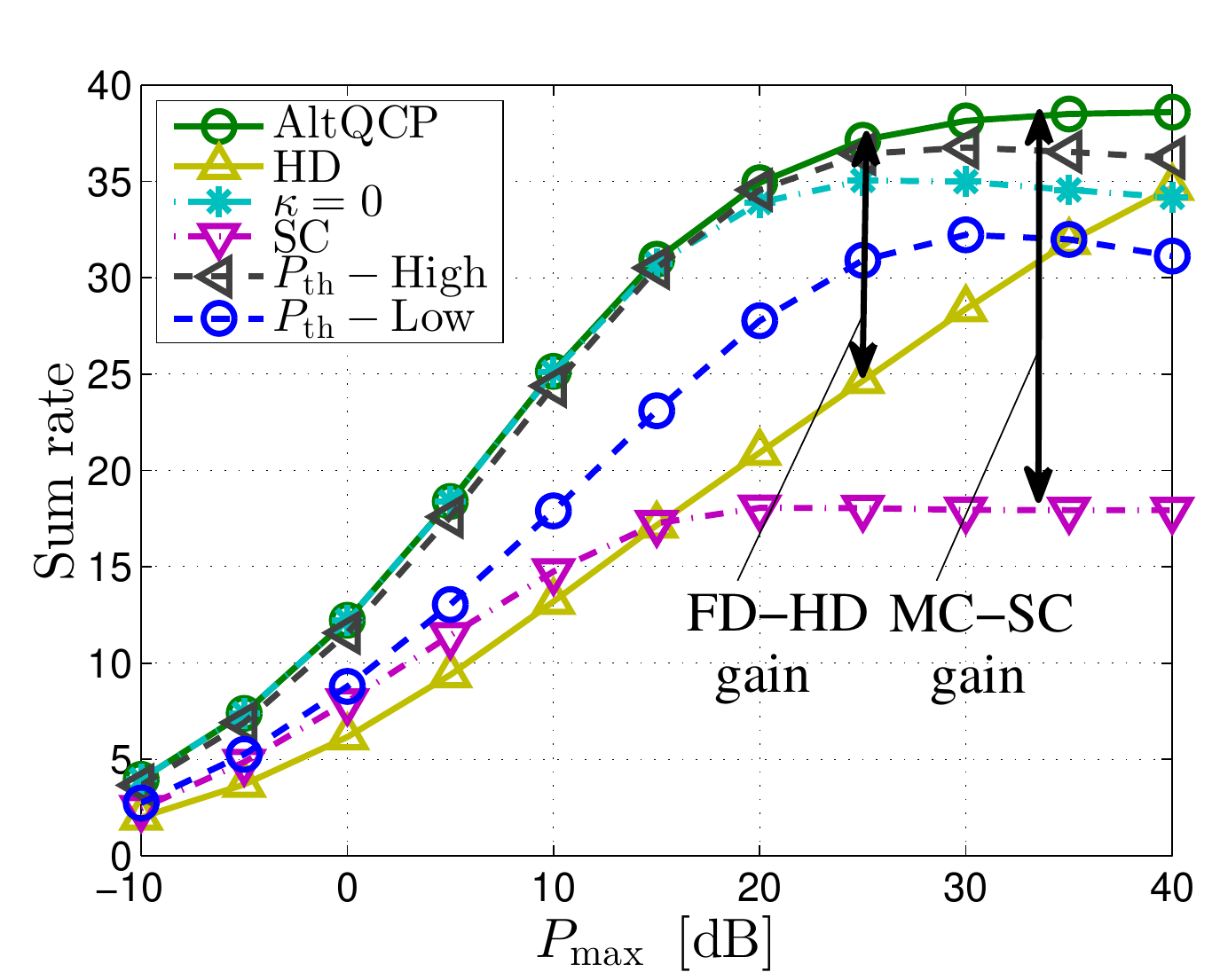}} \label{fig_sr_power}
\hspace{0.4cm}\subfigure[Sum rate vs. hardware inaccuracy]{\includegraphics[width = 0.666\columnwidth]{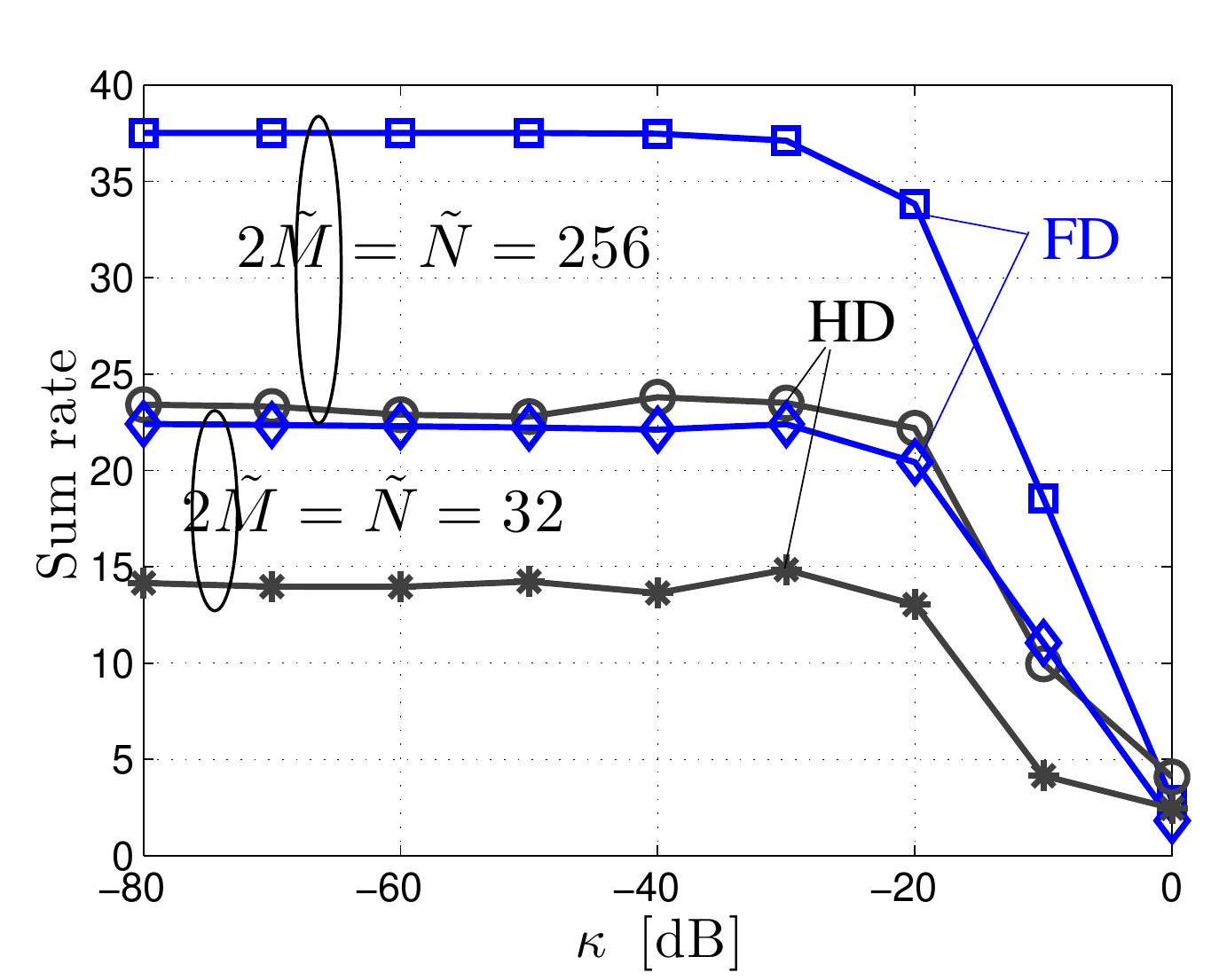}} 
\caption{System performance under various specifications. The application of a distortion-aware design is essential for a system with erroneous hardware or with a high signal-to-noise ratio (SNR). The default parameter set is $\kappa= - 50~\text{dB}$, $\sigma_{\text{n}}^2 := \sigma_{i,k}^2 = -30~\text{dB}$ for (i)-(l).  } \label{fig:MSE_sumrate}
\vspace{-0pt}
\end{figure*}
In this section we evaluate the behavior of the studied FD MC system via numerical simulations. In particular, we evaluate the proposed designs in Sections~\ref{sec:WMMSE} and~\ref{sec:WMMSE_CSI_Error} for various system situations, and under the impact of transceiver inaccuracy and CSI error. Communication channels $\ma{H}_{ii}^k$ follow an uncorrelated Rayleigh flat fading model with variance $\rho$. For the SI channel we follow the characterization reported in \cite{FD_ExperimentDrivenCharact}, \rev{indicating a Rician distribution for the SI channel}. In this respect we have $\ma{H}_{ij} \sim \mathcal{CN}\left( \sqrt{\frac{\rho_{\text{si}} K_R}{1+K_R}} \ma{H}_0 , \frac{\rho_{\text{si}}}{1+K_R} \ma{I}_{M_{i}} \otimes \ma{I}_{N_{j}} \right)$ where $\rho_{\text{si}}$ represents the SI channel strength, $\ma{H}_0$ is a deterministic term,\footnote{For simplicity, we choose $\ma{H}_0$ as a matrix of all-$1$ elements.} and $K_R$ is the Rician coefficient. For each channel realization, the resulting performance is evaluated by employing different design strategies and for various system parameters. The overall system performance is then averaged over $100$ channel realizations. Unless otherwise is stated, the following values are used to define our default setup: $K = 4$, $K_R = 10$, $M:=M_i=N_j=2$, $\rho = -20~\text{dB}$, $\rho_{\text{si}} = 1$, $\sigma_{\text{n}}^2 := \sigma_{i,k}^2 = -30~\text{dB}$, $P_{\text{max}}:=P_i = 1$, $d_i=1$, $\kappa = -30~\text{dB}$ where $\ma{\Theta}_{\text{rx},i} = \kappa \ma{I}_{M_{i}}$ and $\ma{\Theta}_{\text{tx},i} = \kappa \ma{I}_{N_{i}}$, and $\zeta_{ij}^k= -15~\text{dB}$, $\omega_i=1$, $\forall i,j \in \mathbb{I}, \; k \in \mathbb{F}_K$. \par 

\subsection{Algorithm analysis}  \label{sec_AlgorithmAnalysis}

%

Due to the alternating structure, the convergence behavior of the proposed algorithms is of interest, both as a verification for algorithm operation as well as an indication of the algorithm efficiency in terms of the required computational effort. In this part, the performance of AltQCP and AltSDP algorithms are studied in terms of the average convergence behavior and computational complexity. Moreover, the impact of the choice of the algorithm initialization is evaluated. \\
\rev{In Fig.~\ref{fig:init} the average convergence behavior is depicted for different values of $\kappa$~[dB]. In particular, "Min" and "Avg" curves respectively represent the minimum, and the average value of the algorithm objective at the corresponding optimization step over the choice of $20$ random initializations. Moreover, "RSM" represents the right-singular matrix initialization proposed in \cite[Appendix~A]{5585631}. It is observed that the algorithms converge, within $10-30$ optimization iterations, specially as $\kappa$ is small. Although the global optimality of the final solution can not be verified due to the possibility of local solutions, the numerical experiments suggest that the applied RSM initialization shows a better convergence behavior compared to a random initialization. Moreover, it is observed that a higher transceiver inaccuracy results in a slower convergence and a gap with optimality. This is expected, as larger $\kappa$ leads to a more complex problem structure. Note that the algorithm AltQCP shows a smaller value of objective compared to that of AltSDP for any value of $\kappa$, since the impact of CSI error is not considered in the algorithm objective. }  

 
In addition to the algorithm convergence behavior, the required computational complexity is affected by the problem dimension, and the required per-iteration complexity, see Subsection~\ref{alg_complexity}. In Fig.~\ref{fig:Complexity}, the required computation time (CT) is depicted for different number of antennas, as well as different number of subcarriers\footnote{The reported CT is obtained using an Intel Core i$5-3320$M processor with the clock rate of $2.6$ GHz and $8$ GB of random-access memory (RAM). As our software platform we have used MATLAB $2013$a, on a $64$-bit operating system.}. It is observed that the AltSDP results in a significantly higher CT, compared to AltQCP. This is expected as the consideration of CSI error in AltSDP results in a larger problem dimension, and hence higher complexity. Moreover, the obtained closed-form solution expressions in AltQCP result in a more efficient implementation. Nevertheless, the required CT for AltQCP is still higher than the threshold-based low-complexity approaches, see Subsection~\ref{Sim_benchmarks}, due to the expanded problem dimension associated with the impact of \rev{RSI} and \rev{ICL}.    

\subsection{Performance comparison}
In this part we evaluate the performance of the proposed AltSDP and AltQCP algorithms in terms of the resulting worst-case MSE, see Subsection~\ref{WC_CSI}, under various system conditions.  

\subsubsection{Comparison benchmarks} \label{Sim_benchmarks}
In order to facilitate a meaningful comparison, we consider popular approaches for the design of FD single-carrier bidirectional systems, or the available designs for other MC systems with simplified assumptions, see Subsection~\ref{related_works}. The following approaches are hence implemented as our evaluation framework:    
\begin{itemize}[leftmargin=*]
\item \textit{AltSDP}:~The AltSDP algorithm proposed in Section~\ref{sec:WMMSE_CSI_Error}. The impact of the hardware distortions leading to inter-carrier leakage, as well as CSI error are taken into account.
\item \textit{AltQCP}:~The AltQCP algorithm proposed in Section~\ref{sec:WMMSE}. The algorithm operates on the simplified assumption that the CSI error does not exist, i.e., $\zeta=0$, and hence focuses on the impact of hardware distortions.   
\item \textit{HD}:~The AltSDP algorithm is used on an equivalent HD setup, where the communication directions are separated via a time division duplexing (TDD) scheme. 
\rev{\item \textit{$\kappa=0$}:~The impact of CSI error is taken into account similar to, e.g., \cite{ZTH:132, ng2012dynamic}. Nevertheless the impact of hardware distortion, leading to inter-carrier leakage, is ignored.     
\item \textit{SC}:~The optimal single carrier design applied to the defined MC system, following a similar approach as in \cite{DMBS:12, CRYL:15}. The impact of CSI error and hardware distortions are taken into account.}   
\end{itemize}
\rev{Other than the approaches that directly deal with the impact of \rev{RSI}, e.g., \cite{DMBS:12, CRYL:15}, a low complexity design framework is proposed in \cite{ Huberman2014, JTLH:12}, by introducing an interference power threshold, denoted as $P_{\text{th}}$. In this approach, it is assumed that the \rev{SI} signal can be perfectly subtracted, given the \rev{SI} power is kept below $P_{\text{th}}$. In this regard, we evaluate the extended version of \cite{JTLH:12} on the defined MC setup for three values of $P_{\text{th}}$: 
\begin{itemize} [leftmargin=*]
\item \textit{$P_{\text{th}}-\{\infty, \text{High}, \text{Low}\}$}:~representing a design by respectively choosing $P_{\text{th}}=$ $\infty$, $P_i$, $P_i/10$, representing a system with perfect, high, and low dynamic range conditions. 
\end{itemize}}

\subsubsection{Visualization}

In Figs.~\ref{fig:MSE_sumrate}~(a)-(h) the average performance of the defined benchmark algorithms in terms of the worst-case (WC) MSE are depicted. The average sum rate behavior of the system is depicted in Fig.~\ref{fig:MSE_sumrate}~(i)-(l). \\ 

In Fig.~\ref{fig:MSE_sumrate}~(a) the impact of transceiver inaccuracy is depicted on the resulting WC-MSE. It is observed that the estimation accuracy is degraded as $\kappa$ increases. For the low-complexity algorithms, where the impact of hardware distortion is not considered, the resulting MSE goes to infinity as $\kappa$ increases. Nevertheless, the resulting MSE reaches a saturation point for the distortion-aware algorithms, i.e, AltSDP and AltQCP. This is since for the data streams affected with a large distortion intensity, the decoder matrices are set to zero which limits the resulting MSE to the magnitude of the data symbols. Moreover, the AltSDP method outperforms the other performance benchmarks for all values of $\kappa$. It is worth mentioning that the significant gain of an FD system with low $\kappa$ over the HD counterpart, disappears for a larger levels of hardware distortion where AltSDP and HD result in a close performance.    

In Fig.~\ref{fig:MSE_sumrate}~(b) the impact of the CSI error is depicted. It is observed that the estimation MSE increases for a larger value of $\zeta$. For the low-complexity algorithms where the impact of CSI error is not considered, the resulting MSE goes to infinity, as $\zeta$ increases. Nevertheless, the performance of the AltSDP method saturates by choosing zero decoder matrices, following a similar concept as for Fig.~\ref{fig:MSE_sumrate}~(a). It is observed that the performance of the AltSDP and AltQCP methods deviate as $\zeta$ increases, however, they obtain a similar performance for a small $\zeta$. Similar to Fig.~\ref{fig:MSE_sumrate}~(a), a significant gain is observed in comparison to the HD and SC cases, for a system with accurate CSI.  

In Fig.~\ref{fig:MSE_sumrate}~(c) the impact of the thermal noise variance is depicted. It is observed that the resulting performance degrades for the distortion-aware algorithms, as the noise variance increases. Nevertheless, we observe a significant performance degradation for the threshold-based algorithms, particularly $P_{\text{th}}-\text{Low}$, in the low noise regime. This is since the imposed interference power threshold tends to reduce the transmit power, which results in a larger decoder matrices in a low-noise regime. This, in turn, results in an increased impact of distortion. Nevertheless, as the noise variance increases, the algorithm chooses decoding matrices with a smaller norm in order to reduce the impact of noise. This also reduces the impact of hardware distortions. Similar to Fig.~\ref{fig:MSE_sumrate}~(a), the proposed AltSDP method outperforms the other comparison benchmarks. It is observed that the performance degradation caused by ignoring the CSI error in AltQCP, or by applying a simplified single carrier design, is significant particularly for a system with a small noise variance. 

In Fig.~\ref{fig:MSE_sumrate}~(d) the impact of the communication channel strength is observed on the resulting system performance. It is observed that the MSE decreases in most parts as the communication channel becomes stronger. Nevertheless, the system performance saturates, due to the impact of hardware distortion which increase proportional to the transmit/receive power at each chain. Moreover, the performance of the methods with a perfect hardware/CSI assumption saturates at a higher MSE, due to the impact of the ignored effect. Moreover, the algorithms AltQCP and AltSDP result in an approximately similar performance for a system with a high channel strength. This is since for a high $\rho$ regime, the impact of thermal noise and CSI error become less significant. As a result the system performance is dominated by the impact of distortion which is amplified due to the higher channel strength.

 
In Fig.~\ref{fig:MSE_sumrate}~(e) the impact of the number of subcarriers is observed on the resulting MSE. It is observed that a higher number of subcarriers result in a higher error for all benchmark methods. This is expected as a higher number of subcarriers enables a higher number of communication streams, resulting in a lower available per-stream power. The performance of the SC design reaches optimality of a single carrier system, as expected. Nevertheless the performance of the SC scheme deviates from optimality as $K$ increases, and results in the highest MSE in comparison to the evaluated benchmarks, for $K \geq 5$. This is expected, as higher independent subcarriers represent a channel with a higher frequency selectivity which calls for a specialized MC design.   

In Fig.~\ref{fig:MSE_sumrate}~(f) the impact of the number of antennas is observed. As expected, a higher number of antennas results in an increased performance for all of the performance benchmarks. In particular, a higher number of antennas enables the system to better overcome the CSI error, for a fixed $\zeta$, and also to direct the transmit power in the desired channel and not in the self-interference path.

\rev{In Fig.~\ref{fig:MSE_sumrate}~(g) the impact of the accuracy of transmit and receiver chains are studied, where $\kappa \text{[dB]} + \beta \text{[dB]}= \mathcal{K}_{\text{sum}}$, i.e., the sum-accuracy (in dB scale) is fixed. For instance, for an FD transceiver with massive antenna arrays where the utilization of analog cancelers is not feasible, and also the quantization bits are considered as costly resources, the value of $\mathcal{K}_{\text{sum}}$ is related to the total number of quantization bits. The similar evaluation regarding the number of transmit/receive antennas is performed in Fig.~\ref{fig:MSE_sumrate}~(h), where $M_{\text{t}} + M_{\text{r}} =  \mathcal{M}_{\text{sum}}$. It is observed that different available resources, i.e., $\mathcal{K}_{\text{sum}}$, $\mathcal{M}_{\text{sum}}$, result in different optimal allocations. However, as a general insight, it is observed that the performance is degraded when resources are concentrated only on transmit or receive side. Please note that similar approach can be used for evaluating different cost models for accuracy and antenna elements, regarding different system setups, or SIC specifications.}

In Fig.~\ref{fig:MSE_sumrate}~(i)-(l) the average sum rate behavior of the system depicted. In Fig.~\ref{fig:MSE_sumrate}~(i), the impact of hardware inaccuracy is depicted. It is observed that a higher $\kappa$ results in a smaller sum rate. Moreover, the obtained gains via the application of the defined MC design in comparison to the designs with frequency-flat assumption, and via the application of FD setup in comparison to HD setup, are evident for a system with accurate hardware conditions. Conversely, it is observed that a design with consideration of hardware impairments is essential as $\kappa$ increases. In Fig.~\ref{fig:MSE_sumrate}~(j) and (k), the opposite impacts of noise level, and the maximum transmit power are observed on the system sum rate. It is observed that the system sum rate increases as noise level decreases, or as the maximum transmit power increases\footnote{\rev{For the algorithms with a zero-distortion assumption, the maximum allowed transmit power is utilized to reduce the impact of thermal noise. However, this result in an amplified distortion effect, and a reduced performance as SNR increases, also see the MSE peaks in Fig.~\ref{fig:MSE_sumrate}~(c)-(d) for the same algorithms.}}. In both cases, the gain of AltQCP method, in comparison to the methods which ignore the impact of hardware distortions are observed for a high SNR conditions, i.e., for a system with a high transmit power or a low noise level. \rev{In Fig~\ref{fig:MSE_sumrate}~(l) the performance of the asymmetric setup, studied is Section~IV, is depicted, assuming $|\mathbb{I}|=5$, and $d_i=1$. It is observed that the gain of FD system (over the HD counterpart) vanishes rapidly as the hardware distortions increase.}

\vspace{-4mm}
\section {Conclusion} \label{sec:conclusion}
The application of bi-directional FD communication presents a potential for improving the spectral efficiency. Nevertheless, such systems are limited due to the impact of residual self-interference. This issue becomes more crucial in a multi-carrier system, where the residual self-interference spreads over multiple carriers, due to the impact of hardware distortion. In this work we have presented a modeling and design framework for \rev{an FD MIMO} OFDM system, taking into account the impact of hardware distortions leading to inter-carrier leakage, as well as the impact of CSI error.

It is observed that the application of a distortion-aware design is essential, as transceiver accuracy degrades, and inter-carrier leakage becomes a dominant factor. Moreover, a significant gain is observed compared to the usual single-carrier approaches, for a channel with frequency selectivity. However, the aforementioned improvements are obtained at the expense of a higher design computational complexity. 

\vspace{-3mm}
\appendix 
\rev{We start the proof with the characterization of the impact of distortion on the transmit chains. The proof to the receiver characterization is obtained similarly. The statistical independence properties at the frequency domain directly follows from the time domain statistical independence ${e}_{\text{t},l} (t) \bot {v}_{l} (t)$, and~${e}_{\text{t},l} (t) \bot {e}_{\text{t},l^{'}} (t)$, and the linear nature of the transformation in (\ref{FD_P2P_FrequencyDomain_Tx}). The Gaussian and zero-mean properties similarly follow for ${e}_{\text{t},l}^k$ as a linearly weighted sum of the zero-mean Gaussian values ${e}_{\text{t},l}(m T_s)$. The variance of ${e}_{\text{t},l}^k$ can be hence obtained as 
\begin{align} 
& \mathbb{E}\left\{ \left|{e}_{\text{t},l}^k \right|^2  \right\}  = \mathbb{E}\Bigg\{ \frac{1}{K}\left( \sum_{m=0}^{K-1} e_{\text{t},l} (m T_{\text{s}})  e^{-\frac{j2 \pi mk}{K}} \right)   \nonumber \\ & \;\;\;\;\;\;\;\;\;\;\;\;\;\;\;\;\;\;\;\;\;\;\;\;\;\;\;\;\;\;\; \times \left(\sum_{n=0}^{K-1} {e_{\text{t},l}^{*} (n T_{\text{s}})} e^{\frac{j2 \pi nk}{K}}\right)  \Bigg\} \label{appendix_P2P_a} \\
& = \frac{1}{K} \sum_{m=0}^{K-1} \sum_{n=0}^{K-1}  \mathbb{E}\left\{e_{\text{t},l} (m T_{\text{s}})   {e_{\text{t},l}^{*} (n T_{\text{s}}) } \right\} e^{-\frac{j2 \pi (m-n)k}{K}} \label{appendix_P2P_b} \\
& = \kappa_l \mathbb{E}\left\{ \left|{v}_{l} (t) \right|^2  \right\}  \label{appendix_P2P_c}\\
& = \frac{\kappa_l}{K} \sum_{m=1}^{K} \mathbb{E}\left\{ |v_l^m|^2 \right\} \label{appendix_P2P_d}
\end{align}  
where (\ref{appendix_P2P_a}) is obtained via direct application of (\ref{FD_P2P_FrequencyDomain_Tx}), and (\ref{appendix_P2P_c}) is obtained from (\ref{eq_model_distortion_stat_1}), and the statistical independence of $e_{\text{t},l}$ at the subsequent time samples from (\ref{eq_model_distortion_stat_2}). The identity (\ref{appendix_P2P_d}) follows from the Parseval's theorem on the energy conversation over orthonormal Fourier basis.     }

\end{document}